\newcommand*{\QUANTUM}{}%
\ifdefined\QUANTUM
\documentclass[a4paper,onecolumn,11pt,accepted=2020-11-06]{quantumarticle}
\pdfoutput=1
\else
\documentclass[11pt]{article}
\fi

\usepackage{geometry}
\usepackage{amsmath}
\usepackage{graphicx}
\usepackage{dcolumn}
\usepackage{epstopdf}
\usepackage{bm}
\usepackage{subfig}
\usepackage{braket}
\usepackage{algorithm}
\usepackage{algorithmic}
\usepackage[sort,square,numbers]{natbib}
\usepackage{adjustbox}
\usepackage{cellspace} 
\usepackage{makecell} 
\setcellgapes{3pt}
\usepackage{qcircuit}

\usepackage[T1]{fontenc}
\usepackage{color}
\usepackage{amsthm}
\usepackage{amssymb}
\usepackage{esint}
\usepackage{comment}
\usepackage{hyperref}
\usepackage[capitalize]{cleveref}

\makeatletter

\newtheorem{thm}{\protect\theoremname}
\theoremstyle{plain}
\newtheorem{lem}[thm]{\protect\lemmaname}
\theoremstyle{plain}
\newtheorem{rem}[thm]{\protect\remarkname}
\theoremstyle{plain}
\newtheorem*{lem*}{\protect\lemmaname}
\theoremstyle{plain}

\theoremstyle{plain}
\newtheorem{cor}[thm]{\protect\corollaryname}

\newtheorem{manualtheoreminner}{Theorem}
\newenvironment{manualtheorem}[1]{%
  \renewcommand\themanualtheoreminner{#1}%
  \manualtheoreminner
}{\endmanualtheoreminner}

\usepackage[USenglish]{babel}
  \providecommand{\corollaryname}{Corollary}
  \providecommand{\lemmaname}{Lemma}
  \providecommand{\propositionname}{Proposition}
  \providecommand{\remarkname}{Remark}
\providecommand{\theoremname}{Theorem}
\usepackage[normalem]{ulem}
\usepackage{amstext}

\begin{document}

\global\long\def\ve{\varepsilon}
\global\long\def\R{\mathbb{R}}
\global\long\def\Rn{{\mathbb{R}^{n}}}
\global\long\def\Rd{{\mathbb{R}^{d}}}
\global\long\def\E{\mathbb{E}}
\global\long\def\P{\mathbb{P}}
\global\long\def\bx{\mathbf{x}}
\global\long\def\vp{\varphi}
\global\long\def\ra{\rightarrow}
\global\long\def\smooth{C^{\infty}}
\global\long\def\symm{\mathcal{S}^n}
\global\long\def\psd{\mathcal{S}^n_{+}}
\global\long\def\pd{\mathcal{S}^n_{++}}
\global\long\def\dom{\mathrm{dom}\,}
\global\long\def\intdom{\mathrm{int}\,\mathrm{dom}\,}
\global\long\def\Tr{\mathrm{Tr}}

\newcommand{\LL}[1]{\textcolor{blue}{[LL:#1]}}
\newcommand{\YT}[1]{\textcolor{cyan}{[YT:#1]}}
\newcommand{\REV}[1]{#1}

\newcommand{\bvec}[1]{\mathbf{#1}}
\renewcommand{\Re}{\mathrm{Re}}
\renewcommand{\Im}{\mathrm{Im}}
\newcommand{\textred}[1]{\textcolor{red}{#1}}

\newcommand{\mc}[1]{\mathcal{#1}}
\newcommand{\mf}[1]{\mathfrak{#1}}
\newcommand{\mcV}{\mathcal{V}}
\newcommand{\Vin}{V_{\mathrm{in}}}
\newcommand{\Vstar}{V^{\ast}}
\newcommand{\Jstar}{J_{\ast}}
\newcommand{\tJstar}{\wt{J}_{\ast}}
\newcommand{\Vout}{V_{\mathrm{out}}}
\newcommand{\RPA}{\mathrm{RPA}}
\newcommand{\xc}{\mathrm{xc}}
\newcommand{\vr}{\bvec{r}}
\newcommand{\vF}{\bvec{F}}
\newcommand{\vg}{\bvec{g}}
\newcommand{\vR}{\bvec{R}}
\newcommand{\vq}{\bvec{q}}
\newcommand{\vx}{\bvec{x}}
\newcommand{\ud}{\,\mathrm{d}}
\newcommand{\ext}{\mathrm{ext}}
\newcommand{\KS}{\mathrm{KS}}
\newcommand{\Exc}{E_{\mathrm{xc}}}
\newcommand{\Vxc}{\hat{V}_{\mathrm{xc}}}
\newcommand{\Vion}{\hat{V}_{\mathrm{ion}}}
\newcommand{\abs}[1]{\lvert#1\rvert}
\newcommand{\norm}[1]{\lVert#1\rVert}
\newcommand{\average}[1]{\left\langle#1\right\rangle}
\newcommand{\wt}[1]{\widetilde{#1}}
\newcommand{\hxc}{\mathrm{hxc}}

\newcommand{\etc}{\textit{etc.}~}
\newcommand{\etal}{\textit{et al}}  
\newcommand{\ie}{\textit{i.e.}~}
\newcommand{\eg}{\textit{e.g.}~}
\newcommand{\Or}{\mathcal{O}}
\newcommand{\mcF}{\mathcal{F}}
\newcommand{\lmin}{\lambda_{\min}}
\newcommand{\lmax}{\lambda_{\max}}
\newcommand{\Ran}{\text{Ran}}
\newcommand{\I}{\mathrm{i}} 
\newcommand{\EE}{\mathbb{E}}
\newcommand{\NN}{\mathbb{N}}
\newcommand{\RR}{\mathbb{R}}
\newcommand{\CC}{\mathbb{C}}
\newcommand{\ZZ}{\mathbb{Z}}
\newcommand{\Hper}{H^1_\#(\Omega)}
\newcommand{\jmp}[1]{\jl#1\jr}
\newcommand{\al}{\{\hspace{-3.5pt}\{}
\newcommand{\avg}[1]{\al#1\ar}
\newcommand{\jl}{[\![}
\newcommand{\jr}{]\!]}
\newcommand{\VN}{\mathbb V_N}
\newcommand{\angstrom}{\mbox{\normalfont\AA}~}
\newcommand{\psucc}{p_\text{success}}

\title{Optimal \REV{polynomial based} quantum eigenstate filtering with application to solving quantum linear systems}

\author{Lin Lin}
        \affiliation{Department of Mathematics, University of California, Berkeley,  CA 94720, USA}
        \affiliation{Computational Research Division, Lawrence Berkeley National Laboratory, Berkeley, CA 94720, USA}
        \orcid{0000-0001-6860-9566}

\author{Yu Tong}
        \affiliation{Department of Mathematics, University of California, Berkeley,  CA 94720, USA}
        \orcid{0000-0002-7555-9373}

\begin{abstract}

We present a quantum eigenstate filtering algorithm based on quantum signal processing (QSP) and minimax polynomials. The algorithm allows us to efficiently prepare a target eigenstate of a given Hamiltonian, if we have access to an initial state with non-trivial overlap with the target eigenstate and have a reasonable lower bound for the spectral gap. We apply this algorithm to the quantum linear system problem (QLSP), and present two algorithms based on quantum adiabatic computing (AQC) and quantum Zeno effect respectively. Both algorithms prepare the final solution as a pure state, and achieves the near optimal $\mathcal{\widetilde{O}}(d\kappa\log(1/\epsilon))$ query complexity for a $d$-sparse matrix,
where $\kappa$ is the condition number, and $\epsilon$ is the desired precision.  Neither algorithm uses phase estimation or amplitude amplification.

\end{abstract}

\maketitle

\section{Introduction}

Eigenvalue problems have a wide range of applications in scientific and engineering computing. Finding ground states and excited states of quantum many-body  Hamiltonian operators, Google's PageRank algorithm, and principle component analysis are just a few prominent examples.  Some problems that are not apparently eigenvalue problems may benefit from a reformulation into eigenvalue problems. One noticeable example is the quantum linear systems problem (QLSP), which aims at preparing a state that is proportional to the solution of a given linear system, i.e. $\ket{x}=A^{-1}\ket{b}/\norm{A^{-1}\ket{b}}$ 
on a quantum computer \REV{($\|\cdot\|$ denotes the vector 2-norm)}. Here $A\in\CC^{N\times N}$, and $\ket{b}\in\CC^N$. \REV{We give a more detailed definition of the QLSP in Section~\ref{sec:qlsp_aqc}.}
All QLSP solvers share the desirable property that  the complexity with respect to the matrix dimension can be as low as $\Or(\operatorname{polylog}(N))$, which is exponentially faster compared to  known classical solvers. Due to the wide applications of linear systems, the efficient solution of QLSP has received significant attention in recent years 
\cite{HarrowHassidimLloyd2009,ChildsKothariSomma2017,GilyenSuLowEtAl2019,SubasiSommaOrsucci2019,AnLin2019,chakraborty2018power,WossnigZhaoPrakash2018,CaoPapageorgiouPetrasEtAl2013,XuSunEndoEtAl2019,Bravo-PrietoLaRoseCerezoEtAl2019}. By reformulating QLSP into an eigenvalue problem, recent developments have yielded near-optimal query-complexity with respect to $\kappa$ (the condition number of $A$\REV{, defined as the ratio between the largest and the smallest singular value of $A$, or $\kappa=\norm{A}\norm{A^{-1}}$}) \cite{SubasiSommaOrsucci2019,AnLin2019}, which is so far difficult to achieve using alternative methods.

Consider a Hermitian matrix $H\in\CC^{N\times N}$, which has a known interior eigenvalue $\lambda$ separated from the rest of the spectrum by a gap (or a lower bound of the gap) denoted by $\Delta$. Let $P_{\lambda}$ be the spectral projector associated with the eigenvalue $\lambda$. The goal of the  quantum eigenstate filtering problem is to find a certain smooth function $f(\cdot)$, so that $\norm{f(H-\lambda I)-P_{\lambda}}$ is as small as possible, and  there should be a unitary quantum circuit $U$ that efficiently implements $f(H-\lambda I)$.
Then given an initial state $\ket{x_0}$ so that $\norm{P_{\lambda} \ket{x_0}}=\gamma>0$, $f(H-\lambda I)\ket{x_0}$ filters out the unwanted 
spectral components in $\ket{x_0}$ and is approximately an eigenstate of $H$ corresponding to $\lambda$. We assume that  $H$ can be block-encoded into a unitary matrix $U_H$~\cite{GilyenArunachalamWiebe2019}, which is our input model for $H$ and requires a certain amount of ancilla qubits. 
The initial state is prepared by an oracle $U_{x_0}$. In this paper when comparing the number of qubits needed, we focus on the \textit{extra} ancilla qubits introduced by the various methods used, which exclude the ancilla qubits used in the block-encoding of $H$.

In this paper, we develop a polynomial-based filtering method, which chooses $f=P_{\ell}$ to be a $\ell$-th degree polynomial. We prove that our choice yields the \textit{optimal} compression ratio among all polynomials.  Assume that the information of $H$ can be accessed through  its block-encoding. Then we  demonstrate that the optimal eigenstate filtering polynomial can be efficiently implemented using the recently developed quantum signal processing (QSP) \cite{GilyenSuLowEtAl2019,LowChuang2017}, which allows us to implement a general matrix polynomial with a minimal number of ancilla qubits.
More specifically, the query complexity of our method is  $\wt{\mathcal{O}}(1/(\gamma\Delta)\log(1/\epsilon))$ for the block-encoding of the Hamiltonian and $\Or(1/\gamma)$ for initial state preparation, when using amplitude amplification. The number of extra ancilla qubits is merely $3$ when using amplitude amplification, and $2$ when we do not (in this case the $1/\gamma$ factor in both query complexities become $1/\gamma^2$). However in the application to QLSP we can always guarantee $\gamma=\Omega(1)$, and thus not using amplitude amplification only changes the complexity by a constant factor. 

Using the quantum eigenstate filtering algorithm, we present two algorithms to solve QLSP, both achieving a query complexity $\wt{\Or}(\kappa\log(1/\epsilon))$, with constant success probability \REV{(success is indicated by the outcome of measuring the ancilla qubits)}.
For any $\delta>0$, a quantum algorithm that is able to solve a generic QLSP with cost $\Or(\kappa^{1-\delta})$ would imply 
\textsf{BQP}=\textsf{PSPACE}~\cite{HarrowHassidimLloyd2009}. Therefore our algorithm is  near-optimal with respect to $\kappa$ up to a logarithmic factor,  and is optimal with respect to $\epsilon$. The first algorithm (Theorem~\ref{thm:qlsp}) combines quantum eigenstate filtering  with the time-optimal adiabatic quantum computing (AQC) approach \cite{AnLin2019}. We use the time-optimal AQC to prepare an initial state $\ket{x_0}$, which achieves a nontrivial overlap with the true solution as $\gamma=\abs{\braket{x_0|x}}\sim \Omega(1)$. 
Then we apply the eigenstate filtering to $\ket{x_0}$ once, and  the filtered state is  $\epsilon$-close to $\ket{x}$\ upon measurement.
The second algorithm (Theorem~\ref{thm:qlsp_zeno}) combines quantum eigenstate filtering  with the time-optimal version of the approach based on the quantum Zeno effect (QZE)  \cite{BoixoKnillSomma2009,SubasiSommaOrsucci2019}. Instead of preparing one initial vector satisfying $\gamma\sim \Omega(1)$, a sequence of quantum eigenstate filtering algorithm are applied to obtain to the instantaneous eigenstate of interest along an eigenpath. The final state is again  $\epsilon$-close to $\ket{x}$\ upon measurement. Neither algorithm involves phase estimation or any form of amplitude amplification. 
\REV{The first algorithm achieves slightly better dependence on $\kappa$ than the second algorithm, but this comes at the expense of using a time-dependent Hamiltonian simulation procedure \cite{LowWiebe2018} resulting in this algorithm using more ancilla qubits than the second QZE-based algorithm.}
\REV{For both algorithms, because the success is indicated by the outcome of measuring the ancilla qubits, we can repeat the algorithms $\Or(\log(1/\delta))$ times to boost the final success probability from $\Omega(1)$ to $1-\delta$ for arbitrarily small $\delta$.}



\vspace{1em}
\noindent\textbf{Related works:}

A well-known quantum eigenstate filtering algorithm is phase estimation \cite{Kitaev1995}, which relies on Hamiltonian simulation \cite{lloyd1996universal,BerryChildsCleveEtAl2015,LowWiebe2018,LowChuang2017,LowChuang2019,BerryChildsKothari2015} and the quantum Fourier transform. 
We treat the Hamiltonian simulation $e^{-\I H \tau}$ with some fixed $\tau$ as an oracle called $U_{\mathrm{sim}}$, \REV{where $\tau$ satisfies $\tau\|H\|<\pi$}. \REV{Ref. \cite[Appendix B]{GeTuraCirac2019} contains a very detailed analysis of the complexities of using phase estimation together with amplitude amplification. From the analysis in Ref.~\cite{GeTuraCirac2019}}, this approach requires $\wt{\mathcal{O}}(1/(\gamma^2\Delta\epsilon))$ times of queries for $U_{\mathrm{sim}}$, where $\epsilon$ is the target accuracy (the complexity is the same up to logarithmic factors if we use the block-encoding $U_H$ instead of its time-evolution as an oracle);
\REV{the number of queries to the circuit $U_{x_0}$ that prepares the initial trial state is $\wt{\Or}(1/\gamma)$;}
and the number of extra ancilla qubits is  $\Or(\log(1/(\epsilon \Delta)) $. This is non-optimal with respect to both $\gamma$ and $\epsilon$. 

\REV{Several variants of phase estimation are developed to achieve better dependence on the parameters $\gamma$ and $\epsilon$ \cite{PoulinWocjan2009,poulin2009sampling,GeTuraCirac2019}.} The filtering method developed by Poulin and Wocjan \cite{PoulinWocjan2009} (for a task related to eigenstate filtering) improves the query complexities of $U_{\mathrm{sim}}$ and $U_{x_0}$ with respect to $\gamma$ from $\wt{\Or}(1/\gamma^2)$ to $\wt{\Or}(1/\gamma)$.
Ge \etal. \cite[Appendix C]{GeTuraCirac2019}  shows that the method by Poulin and Wocjan can be adapted to the ground state preparation problem so that the  query complexity of $U_{\mathrm{sim}}$ becomes  $\wt{\mathcal{O}}(1/(\gamma\Delta)\log(1/\epsilon)),$ while the complexity of $U_{x_0}$\ remains $\wt{\Or}(1/\gamma)$. The number of extra ancilla qubits is $\Or(\log(1/(\epsilon\Delta))$. \REV{Similar logarithmic dependence on the accuracy in the query complexity has also been achieved in Ref.~\cite{poulin2009sampling}.}

Ge \etal. \cite{GeTuraCirac2019} also proposed two eigenstate filtering algorithms using linear combination of unitaries (LCU) \cite{ChildsKothariSomma2017,BerryChildsKothari2015}, which uses the Fourier basis and the \REV{Chebyshev} polynomial basis, respectively. For both methods, the query complexities for $U_H$ and $U_{x_0}$ are  $\wt{\mathcal{O}}(1/(\gamma\Delta)\log(1/\epsilon))$ and $\wt{\Or}(1/\gamma)$ respectively, and the number of extra ancilla qubits can be reduced to  $\Or(\log\log(1/\epsilon)+\log(1/\Delta))$. The $\log\log(1/\epsilon)$ factor comes from the use of LCU. We remark that these methods were developed for finding the ground state, but can be adapted to compute interior eigenstates as well. Our filtering method has the same query complexity up to polylogarithmic factors. The number of extra ancilla qubits is significantly fewer and does not depend on either $\epsilon$ or $\Delta$, due to the use of QSP.   Our method also uses the optimal filtering polynomial, which solves a minimax problem as recorded in Lemma~\ref{lem:minimax_poly}. There are several other hybrid quantum-classical algorithms to compute ground state energy and to prepare the ground state \cite{StairHuangEvangelista2019, ParrishMcMahon2019}, whose computational complexities are not yet analyzed and therefore we do not make comparisons here.

For solving QLSP, the query complexity of the original Harrow, Hassidim, and Lloyd (HHL) algorithm \cite{HarrowHassidimLloyd2009} scales as $\wt{\Or}(\kappa^2/\epsilon)$, where $\kappa$ is the condition number of $A$, and $\epsilon$ is the target accuracy. Despite the exponential speedup with respect to the matrix dimension, the scaling with respect to $\kappa$ and $\epsilon$ is significantly weaker compared to that in classical methods.  For instance, for positive definite matrices, the complexity of steepest descent (SD) and conjugate gradient (CG) (with respect to both $\kappa$ and $\epsilon$) are only $\Or( {\kappa}\log(1/\epsilon))$ and $\Or( \sqrt{\kappa}\log(1/\epsilon))$, respectively \cite{Saad2003}.

In the past few years, there have been significant progresses towards reducing the pre-constants for quantum linear solvers. In particular, the linear combination of unitary (LCU) \cite{BerryChildsKothari2015,ChildsKothariSomma2017} and quantum signal processing (QSP) or quantum singular value transformation (QSVT) \cite{LowChuang2017,GilyenSuLowEtAl2019} techniques  can reduce the query complexity to $\Or(\kappa^2 \operatorname{polylog}(\kappa/\epsilon))$. Therefore the algorithm is \REV{almost} optimal with respect to $\epsilon$, but is still suboptimal with respect to $\kappa$. The scaling with respect to $\kappa$ can be reduced by the variable-time amplitude amplification (VTAA) \cite{Ambainis2012} technique, and the resulting query complexity for solving QLSP is $\Or(\kappa\operatorname{polylog}(\kappa/\epsilon)))$ \cite{ChildsKothariSomma2017,chakraborty2018power}. However, VTAA requires considerable modification of the LCU or QSP algorithm, and has significant overhead itself. 
To the extent of our knowledge, the performance of VTAA for solving QLSP has not been quantitatively reported in the literature.

  \begin{table}[ht]
  \makegapedcells
    \begin{center}
    \REV{
      \begin{adjustbox}{width=\textwidth} 
        \begin{tabular}{p{6cm}|p{3cm}|p{6cm}}
          \hline
          \hline
          Algorithm &  Query complexity & Remark\\
          \hline
          HHL  \cite{HarrowHassidimLloyd2009}  & $\wt{\Or}(\kappa^2/\epsilon)$ & w. VTAA, complexity becomes $\wt{\Or}(\kappa/\epsilon^3)$ \cite{ambainis2010variable}\\
          \hline
          Linear combination of unitaries (LCU) \cite{ChildsKothariSomma2017,chakraborty2018power}  & $\wt{\Or}(\kappa^2\operatorname{polylog}(1/\epsilon))$ & w. VTAA, complexity becomes $\wt{\Or}(\kappa\operatorname{polylog}(1/\epsilon))$\\
          \hline
          Quantum singular value transformation (QSVT) \cite{GilyenSuLowEtAl2019}  & $\wt{\Or}(\kappa^2\log(1/\epsilon))$ & \\
          \hline
          Randomization method (RM) \cite{SubasiSommaOrsucci2019}   & $\wt{\Or}(\kappa/\epsilon)$ & w. repeated phase estimation, complexity becomes $\wt{\Or}(\kappa\operatorname{polylog}(1/\epsilon))$ \\
          \hline
          Time-optimal adiabatic quantum computing (AQC(exp)) \cite{AnLin2019} & $\wt{\Or}(\kappa\operatorname{polylog}(1/\epsilon))$ & No need for any amplitude amplification. Use time-dependent Hamiltonian simulation.\\
          \hline
          Eigenstate filtering+AQC (Theorem~\ref{thm:qlsp})& $\wt{\Or}(\kappa\log(1/\epsilon))$  & No need for any amplitude amplification. \\
          \hline
          Eigenstate filtering+QZE (Theorem~\ref{thm:qlsp_zeno})& $\wt{\Or}(\kappa\log(1/\epsilon))$  & No need for any amplitude amplification. Does not rely on any complex subroutines.\\
          \hline
          \hline
        \end{tabular}
      \end{adjustbox}
    }
    \end{center}
    \caption{\REV{The number of queries to the block-encoding of the coefficient matrix $A$ for solving QLSP. Some algorithms were not originally formulated using block-encoding as the input model, but can be converted to use the block-encoding model instead. In the HHL algorithm it is assumed that we have access to time-evolution under the Hermitian coefficient matrix as the Hamiltonian. This assumption can be met when we have the block-encoding of $A$ using Hamiltonian simulation technique that results in small overhead \cite{BerryChildsCleveEtAl2015,LowWiebe2018,LowChuang2017,LowChuang2019,BerryChildsKothari2015}. The LCU method \cite{ChildsKothariSomma2017} and the gate-based implementation of the RM method \cite{SubasiSommaOrsucci2019} both assume oracles to access elements of $A$. However in both cases the oracles lead to a block-encoding $A$ which can be used in the algorithms. The same can be said of the sparse-access input model in Ref.~\cite{chakraborty2018power}. Time complexities and gate complexities are converted to query complexities with respect to the oracles in this paper. \REV{\cite[Thereom 41]{GilyenSuLowWiebe2018Long} gives the implementation of the pseudoinverse using QSVT. This can be used to solve the QLSP by applying this pseudoinverse to the quantum state representing the right-hand side.} }  }
    \label{tab:compare_qlsp_algs}
  \end{table}

The recently developed randomization method (RM) \cite{SubasiSommaOrsucci2019} is the first algorithm that yields near-optimal scaling with respect to $\kappa$, without using techniques such as VTAA. RM was inspired by adiabatic quantum computation (AQC) \cite{FarhiGoldstoneGutmannEtAl2000,AlbashLidar2018,JansenRuskaiSeiler2007}, but relies on the quantum Zeno effect.  Both RM and AQC reformulate QLSP into an eigenvalue problem. The runtime complexity of RM is  $\Or(\kappa\log (\kappa)/\epsilon)$. The recently developed time-optimal AQC(p) and AQC(exp) approaches \cite{AnLin2019} reduces the runtime complexity to $\Or(\kappa/\epsilon)$ and $\Or(\kappa \operatorname{polylog}(\kappa/\epsilon))$, respectively. In particular, AQC(exp) achieves the near-optimal complexity with respect to both $\kappa$ and $\epsilon$, without relying on any amplification procedure. We also remark that numerical observation indicate that the time complexity of the quantum approximate optimization algorithm (QAOA) \cite{farhi2014quantum} can be \REV{as low as} $\Or(\kappa \operatorname{polylog}(1/\epsilon))$ \cite{AnLin2019}. The direct analysis of the complexity of QAOA without relying on the complexity of adiabatic computing (such as AQC(exp)) remains an open question.
We demonstrate that quantum eigenstate filtering provides a more versatile approach to obtain the near optimal complexity for solving QLSP. In particular, it can be used to reduce the complexity with respect to $\epsilon$ for both adiabatic computing and quantum Zeno effect based methods. 
\REV{In Table~\ref{tab:compare_qlsp_algs} we compare these aforementioned algorithms in terms of the number of queries to the block-encoding of $A$. We note that these algorithms rely on different input models but they can all be slightly modified to use the block-encoding assumed in this work.}

\REV{
Recently quantum-inspired classical algorithms based on $\ell^2$-norm sampling assumptions \cite{Tang2018quantum, Tang2019quantum} have been developed that are only up to polynomially slower than the corresponding quantum algorithms. Similar techniques have been applied to solve low-rank linear systems \cite{ChiaLinWang2018, GilyenLloydTang2018}, which achieve exponential speedup in the dependence on the problem size compared to the traditional classical algorithms for the same problem. However, it is unclear whether the classical $\ell^2$-norm sampling can be done efficiently without access to a quantum computer in the setting of this work. The quantum-inspired classical algorithms also suffer from many practical issues making their application limited to highly specialized problems \cite{ArrazolaEtAl2019}. Most importantly, the assumption of low-rankness is crucial in these algorithms. Our work is based on the block-encoding model, which could be used to efficiently represent low-rank as well as full-rank matrices on a quantum computer.}

 
\vspace{1em}
\noindent\textbf{Notations:} 
\REV{In this paper we use the following asymptotic notations besides the usual $\Or$ notation: we write $f=\Omega(g)$ if $g=\Or(f)$; $f=\Theta(g)$ if $f=\Or(g)$ and $g=\Or(f)$; $f=\wt{\Or}(g)$ if $f=\Or(g\operatorname{polylog}(g))$. }

\REV{We use $\|\cdot\|$ to denote vector or matrix 2-norm: when $v$ is a vector we denote by $\|v\|$ its 2-norm, and when $A$ is matrix we denote by $\|A\|$ its operator norm. For two quantum states $\ket{x}$ and $\ket{y}$, we sometimes write $\ket{x,y}$ to denote $\ket{x}\ket{y}$. We use fidelity to measure how close to each other two quantum states are. Note there are two common definitions for the fidelity between two pure states $\ket{\phi}$ and $\ket{\varphi}$: it is either $|\braket{\phi|\varphi}|$ or $|\braket{\phi|\varphi}|^2$. Throughout the paper we use the former definition.}

\vspace{1em}
\noindent\textbf{Organization:}
The rest of the paper is organized as follows. In Section~\ref{sec:block_encode_qsp} we briefly review block-encoding and QSP, as well as using QSP to directly solve QLSP with a non-optimal complexity. In Section~\ref{sec:eigstate_filtering} we introduce the minimax polynomial we are using and our eigenstate filtering method based on it. In Section~\ref{sec:qlsp_aqc} we combine eigenstate filtering with AQC to solve the QLSP. In Section~\ref{sec:quantum_zeno} we present another method to solve the QLSP using QZE and eigenstate filtering. In Section~\ref{sec:conclusion} we discuss some practical aspects of our algorithms and future work.


\section{Block-encoding and quantum signal processing}
\label{sec:block_encode_qsp}

For simplicity we assume  $N=2^n$.
An $(m+n)$-qubit unitary operator $U$ is called an $(\alpha, m, \epsilon)$-block-encoding of an $n$-qubit operator $A$, if 
\begin{equation}
\norm{A-\alpha(\bra{0^m}\otimes I) U (\ket{0^m}\otimes I)}\leq \epsilon.
\label{eqn:block_encoding}
\end{equation}
\REV{Another way to express \cref{eqn:block_encoding} is 
\[
U=\begin{pmatrix}
\wt{A}/\alpha & * \\
* & *
\end{pmatrix},
\]
where $*$ can be any block matrices of the correct size and $\|\wt{A}-A\|\leq \epsilon$. For instance, when $m=1$, $\wt{A}/\alpha$ is an $n$-qubit matrix at the upper-left diagonal block of the $(n+1)$-qubit unitary matrix $U$.} \REV{Note that the fact $\wt{A}/\alpha$ is the upper-left block of a unitary matrix implies $\|\wt{A}/\alpha\|\leq\|U\|=1$. Therefore $\|\wt{A}\|\leq\alpha$.}
Many matrices used in practice can be efficiently block-encoded. For instance, if all entries of $A$ satisfies $\abs{A_{ij}}\le 1$, and $A$ is Hermitian and $d$-sparse (i.e. each row / column of $A$ has no more than $d$ nonzero entries), then $A$ has a $(d,n+2,0)$-encoding $U$. \REV{See \cite[Section~4.1]{ChildsKothariSomma2017} and \cite[Lemma 10]{BerryChildsKothari2015} for details, as well as \cite[Lemma 48]{GilyenSuLowWiebe2018Long} for a more general treatment of sparse matrices.}

With a block-encoding available, QSP allows us to construct a block-encoding for an arbitrary polynomial eigenvalue transformation of $A$.  

\begin{thm}
\label{thm:qsp}
 \textbf{\REV{(Polynomial eigenvalue transformation via quantum signal processing\footnote{\REV{Throughout the paper we use the term QSP to refer to this type of polynomial eigenvalue transformation as well.} } \cite[Theorem 31]{GilyenSuLowEtAl2019})}}: Let $U$ be an $(\alpha,m,\epsilon)$-block-encoding of a Hermitian matrix $A$. Let $P\in\RR[x]$ be a degree-$\ell$ real polynomial and $\abs{P(x)}\le 1/2$ for any $x\in[-1,1]$. Then there exists a $(1, m+2,4 \ell \sqrt{\epsilon} / \alpha)$-block-encoding $\wt{U}$ of $P(A/\alpha)$ using $\ell$ queries of $U$, $U^{\dagger}$, and $\Or((m+1)\ell)$ other primitive quantum gates.
\end{thm}

We remark that Theorem~\ref{thm:qsp} does not meet all our needs because of the constraint $\abs{P(x)}\leq 1/2$. This requirement comes from decomposing the polynomial into the sum of an even and an odd polynomial and then summing them up. When $P(x)$ naturally has a parity this requirement becomes redundant. This enables us to get rid of 1 ancilla qubit. Also for simplicity we assume the block-encoding of $A$ is exact. Therefore we have the following theorem, which can be proved directly from \cite[Theorem~2 and Corollary~11]{GilyenSuLowEtAl2019}. 

\begin{manualtheorem}{1'}\REV{\textbf{(Polynomial eigenvalue transformation with definite parity via quantum signal processing)}}
\label{thm:qsp_parity}
Let $U$ be an $(\alpha,m,0)$-block-encoding of a Hermitian matrix $A$. Let $P\in\RR[x]$ be a degree-$\ell$ even or odd real polynomial and $\abs{P(x)}\le 1$ for any $x\in[-1,1]$. Then there exists a $(1, m+1,0)$-block-encoding $\wt{U}$ of $P(A/\alpha)$ using $\ell$ queries of $U$, $U^{\dagger}$, and $\Or((m+1)\ell)$ other primitive quantum gates. 
\end{manualtheorem}


Compared to methods such as LCU, one distinct advantage of QSP is that the number of extra ancilla qubits needed is only $1$ as shown in Theorem~\ref{thm:qsp_parity}. Hence QSP may be possibly carried out efficiently on intermediate-term devices. Furthermore, a polynomial can be expanded into different basis functions as $P(x)=\sum_{k=0}^{\ell} c_k f_k(x)$, where $f_k$ can be the monomial $x^k$, the Chebyshev polynomial $T_k(x)$, or any other polynomial. The performance of LCU crucially
depends on the 1-norm $\norm{c}_1=\sum_{k=0}^\ell |c_k|$, which can be very different depending on the expansion~\cite{ChildsKothariSomma2017}. The block encoding $\wt{U}$ in QSP\ is independent of such a choice, and therefore provides a more \textit{intrinsic} representation of matrix function.
We also remark that in QSP, the construction of the block-encoding $\wt{U}$ involves a sequence of parameters called phase factors. For a given polynomial $P(x)$, the computation of the phase factors can be \REV{efficiently} performed on classical computers \cite{Haah2019,GilyenSuLowWiebe2018Long}.  There are however difficulties in computing such phase factors, which will be discussed in Section~\ref{sec:conclusion}. For simplicity we assume that the phase factors are given and computed without error.

As an example, we demonstrate how to use QSP to solve QLSP \REV{with a Hermitian coefficient matrix $A$, given by its $(\alpha,m,0)$-block-encoding $U_A$. We assume that $A,b$ are normalized as 
\[
\norm{A}=1,\quad \braket{b|b}=1.
\]
We also assume $A$ is Hermitian, \REV{and therefore all the eigenvalues of $A$ are real.} General matrices can be treated using the standard matrix dilation method (see Appendix~\ref{sec:apdx_matrix_dilation}).
Due to the  normalization condition, the block-encoding factor satisfies $\alpha\ge \norm{A}=1$. Furthermore, since $\kappa=\norm{A}\norm{A^{-1}}=\norm{A^{-1}}$,  the smallest singular value of $A$ is $1/\kappa$. Hence the eigenvalues of $A/\alpha$ are contained in the set $[-1/\alpha,-1/(\alpha\kappa)]\cup [1/(\alpha\kappa),1/\alpha]\subseteq \mathcal{D}_{1/(\alpha\kappa)}$, where 
\[
\mathcal{D}_{\delta}:=[-1,-\delta]\cup[\delta,1].
\]
Later we will keep using this notation $\mathcal{D}_{\delta}$ to denote sets of this type.}
 We first find a polynomial $P(x)$ satisfying  $|P(x)|\le 1$ for any $x\in[-1,1]$, and $|P(x)- 1/( c x)|\leq \epsilon'$ on $\mc{D}_{1/(\alpha\kappa)}$ for \REV{$c=4\alpha\kappa/3$}. \REV{Note that $\epsilon'$ is the accuracy of the polynomial approximation, so that the unnormalized state $P(A/\alpha)\ket{b}$ would differ from the desired $(\alpha/c)A^{-1}\ket{b}$ by $\epsilon'$. In order to obtain a normalized solution $P(A/\alpha)\ket{b}/\|P(A/\alpha)\ket{b}\|$ that is $\epsilon$-close to the normalized solution $\ket{x}=A^{-1}\ket{b}/\norm{A^{-1}\ket{b}}$, we first note that $\norm{A^{-1}\ket{b}}\geq 1$. So the normalization would amplify the error by a factor of approximately $c/(\alpha\|A^{-1}\ket{b}\|)\leq 4\kappa/3$. Therefore we may choose $\epsilon'=3\epsilon/4\kappa$. } 
Then we can find an odd polynomial of degree $\Or(\alpha\kappa\log(\kappa/\epsilon))$, where $\epsilon$ is the desired precision, satisfying this by \cite[Corollary 69]{GilyenSuLowWiebe2018Long}.  Then \REV{by Theorem~\ref{thm:qsp_parity}} we have {a circuit $\wt{U}$ satisfying}
$$
\begin{aligned}
\wt{U}\ket{0^{\REV{m+1}}}\ket{b}=&\ket{0^{\REV{m+1}}}(P(A/\alpha)\ket{b})+\ket{\phi}\\
\approx& \ket{0^{\REV{m+1}}}\left(\frac{\alpha}{c}A^{-1}\ket{b}\right)+\ket{\phi},
\end{aligned}
$$
where $\ket{\phi}$ is orthogonal to all states of the form $\ket{0^{\REV{m+1}}}\ket{\psi}$. Measuring the ancilla qubits, we obtain the \REV{a normalized quantum state $P(A/\alpha)\ket{b}/\|P(A/\alpha)\ket{b}\|$ that is \REV{$\epsilon$-}close to the normalized} solution $\ket{x}$ with probability $\Theta\left(\left(\frac{\alpha}{c}\norm{A^{-1}\ket{b}}\right)^2\right)$.

As $\norm{A^{-1}\ket{b}}\geq 1$, the probability of success is $\Omega(1/\kappa^2)$. Using amplitude amplification \cite{BrassardHoyerMoscaEtAl2002}, the number of repetitions needed for success can be improved to $\Or(\kappa)$. Furthermore, the query complexity of application of $\wt{U}$ is $\Or(\alpha\kappa\log(\kappa/\epsilon))$. Therefore the overall query complexity is $\Or(\alpha\kappa^2\log(\kappa/\epsilon))$. 

We observe that the quadratic scaling with respect to $\kappa$ is very much attached to the procedure above: each application of QSP costs $\Or(\kappa)$ queries of $U,U^{\dagger}$, and the other
from that QSP needs to be performed for $\Or(\kappa)$ times. The same argument applies to other techniques such as LCU. To reduce the  $\kappa$ complexity along this line, one must modify the procedure substantially to avoid the multiplication of the two $\kappa$ factors, such as using the modified LCU based on VTAA \cite{ChildsKothariSomma2017}.

\section{Eigenstate filtering using a minimax polynomial}
\label{sec:eigstate_filtering}

Now consider a Hermitian matrix $H$, with a known eigenvalue $\lambda$ that is separated from other eigenvalues by a gap $\Delta$. $H$ is assumed to have an $(\alpha,m,0)$-block-encoding denoted by $U_H$. We want to preserve the $\lambda$-eigenstate while filtering out all other eigenstates. Let $P_\lambda$ denote the projection operator into the $\lambda$-eigenspace of $H$. The basic idea is, suppose we have a polynomial $P$ such that $P(0)=1$ and $|P(x)|$ is small for $x\in\mathcal{D}_{\Delta/(2\alpha)}$, \REV{where we use the notation $\mathcal{D}_\delta=[-1,-\delta]\cup[\delta,1]$ that has been introduced earlier,}  then $P((H-\lambda I)/(\alpha+|\lambda|))\approx P_\lambda$. This is the essence of the algorithm we are going to introduce below.
\REV{The reason we need to introduce the factors $2\alpha$ and $\alpha+|\lambda|$ is that the block-encoding of $H-\lambda I$ will involve a factor $\alpha+|\lambda|$, and this is explained in detail in Appendix~\ref{sec:apdx_block_encoding}. Since $|\lambda|\leq \alpha$ by definition of the operator norm, we have $\alpha+|\lambda|\leq 2\alpha$. Therefore when $\lambda$ is separated from the rest of the spectrum of $H$ by a gap $\Delta$, 0 is separated from the rest of the spectrum of $(H-\lambda I)/(\alpha+|\lambda|)$ by a gap $\Delta/(\alpha+|\lambda|)\geq \Delta/(2\alpha)=\wt{\Delta}$.}

We use the following $2\ell$-degree polynomial
\[
R_{\ell}(x;\Delta)=\frac{T_{\ell}\left(-1+2\frac{x^{2}-\Delta^{2}}{1-\Delta^{2}}\right)}{T_{\ell}\left(-1+2\frac{-\Delta^{2}}{1-\Delta^{2}}\right)},
\]
\REV{where $T_\ell(x)$ is the $\ell$-th Chebysehv polynomial of the first kind. This polynomial is inspired by the shifted and rescaled Chebyshev polynomial discussed in \cite[Theorem 6.25]{Saad2003}. }
A plot of the polynomial is given in Fig.~\ref{minimax_poly}. $R_{\ell}(x;\Delta)$ has several nice properties:

\begin{lem}
\label{lem:minimax_poly}
(i) $R_{\ell}(x;\Delta)$ solves the minimax problem
\[
\underset{p(x)\in\mathbb{P}_{2\ell}[x],p(0)=1}{\mathrm{minimize}}  \max_{x\in\mathcal{D}_{\Delta}}|p(x)|.
\]

(ii) $|R_{\ell}(x;\Delta)|\leq2e^{-\sqrt{2}\ell\Delta}$ for all $x\in\mathcal{D}_{\Delta}$
and $0<\Delta\leq1/\sqrt{12}$. Also $R_{\ell}(0;\Delta)=1$.

(iii) $|R_{\ell}(x;\Delta)|\leq1$ for all $|x|\leq1$.
\end{lem}

\begin{figure}[ht!]
\centering
\includegraphics[width=0.5\textwidth]{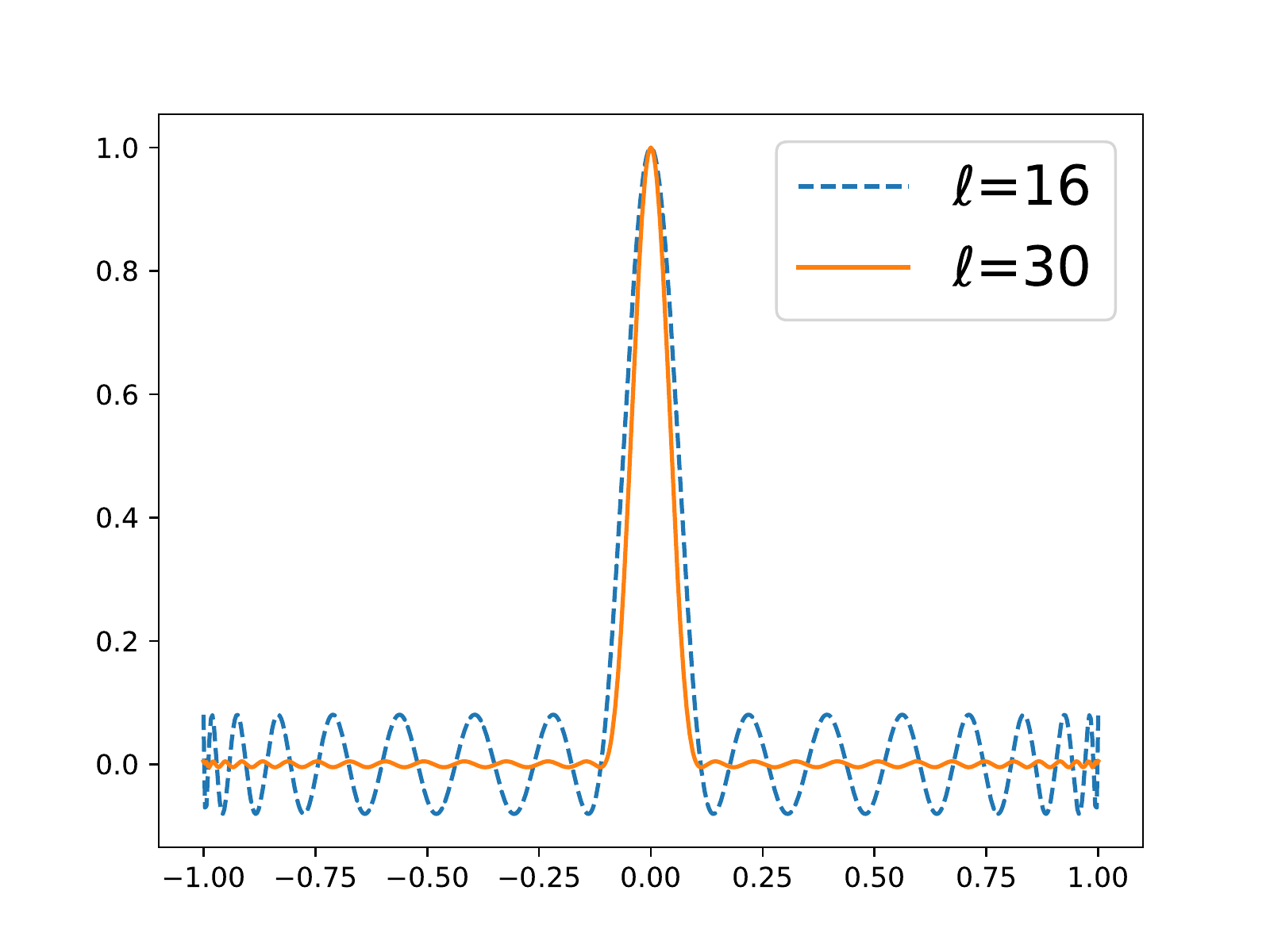}
\caption{The polynomial $R_{\ell}(x,\Delta)$ for $\ell=16$ and $30$, $\Delta=0.1$.}
\label{minimax_poly}
\end{figure}

\REV{A proof of the above lemma is provided in Appendix~\ref{sec:apdx_optimal_poly}.}
If we apply this polynomial to $H-\lambda I$, Lemma~\ref{lem:minimax_poly} (i) states that $R_{\ell}$ achieves the best compression ratio of the unwanted components, among all polynomials of degrees up to $2\ell$. To prepare a quantum circuit, we define $\wt{H}=(H-\lambda I)/(\alpha+|\lambda|)$. \REV{Then} we can also construct a $(1,m+1,0)$-block-encoding for $\wt{H}$ (see Appendix~\ref{sec:apdx_block_encoding}). The gap separating 0 from other eigenvalues of $\wt{H}$ is lower bounded by $\wt{\Delta}=\Delta/2\alpha$, \REV{as explained at the beginning of this section.} \REV{Together with the fact that $\|\wt{H}\|\leq 1$, we find that the spectrum of $\wt{H}$ is contained in $\mathcal{D}_{\wt{\Delta}}\cup\{0\}$.}

\REV{We then apply Lemma~\ref{lem:minimax_poly}. Note that the requirement when $\wt{\Delta}> 1/\sqrt{12}$ might not be satisfied, we can always set $\wt{\Delta}=1/\sqrt{12}$ and this does not affect the asymptotic complexity as $\wt{\Delta}\rightarrow 0$.} Because of (ii) of Lemma~\ref{lem:minimax_poly}, we have 
\[
\|R_\ell(\wt{H},\wt{\Delta})-P_\lambda\|\leq 2e^{-\sqrt{2}\ell\wt{\Delta}}.
\]
Also because of (iii), and the fact that $R_{\ell}(x;\wt{\Delta})$ is even, \REV{we may apply Theorem \ref{thm:qsp_parity} to} implement $R_{\ell}(\wt{H};\wt{\Delta})$ using QSP. This gives the following theorem:

\begin{thm}
\label{thm:eigenstate_filter}
\textbf{(Eigenstate filtering):} Let $H$ be a Hermitian matrix and $U_H$ is an $(\alpha,m,0)$-block-encoding of $H$. If $\lambda$ is an eigenvalue of $H$ that is separated from the rest of the spectrum by a gap $\Delta$, then we can construct a $(1,m+2,\epsilon)$-block-encoding of $P_\lambda$, by $\mathcal{O}((\alpha/\Delta)\log(1/\epsilon))$ applications of (controlled-) $U_H$ and $U^{\dagger}_H$, and $\mathcal{O}((m\alpha/\Delta)\log(1/\epsilon))$ other primitive quantum gates.
\end{thm}

Suppose we can prepare a state $\ket{\psi} = \gamma \ket{\psi_\lambda} + \ket{\perp}$ using an oracle $O_\psi$, where  $\ket{\psi_\lambda}$ is a $\lambda$-eigenvector and $\braket{\psi_\lambda|\perp}=0$, for some $0<\gamma\leq 1$. Theorem \ref{thm:eigenstate_filter} states that we can get an $\epsilon$-approximation to $\ket{\psi_{\lambda}}$ with $\mathcal{O}((\alpha/\Delta)\log(\REV{1/(\gamma\epsilon)}))$ queries to $U_H,$ with a successful application of the block-encoding of $P_\lambda$, denoted by $U_{P_\lambda}$. \REV{The fact we have $1/(\gamma\epsilon)$ instead of $1/\epsilon$ in the logarithm is due to the error amplification going from an unnormalized state to a normalized state, similar to that discussed in  the application of QSP to QLSP in Section~\ref{sec:block_encode_qsp}.} The probability of applying this block-encoding successfully, \ie getting all 0's when measuring ancilla qubits, is at least $\gamma^2$. Therefore when $\ket{\psi}$ can be repeatedly prepared by an oracle, we only need to run $U_{P_\lambda}$ and the oracle on average $\mathcal{O}(1/\gamma^2)$ times to obtain $\ket{\psi_{\lambda}}$ successfully. With amplitude amplification we can reduce this number to $\Or(1/\gamma)$. However this is not necessary when $\gamma=\Omega(1)$, when without using amplitude amplification we can already obtain $\ket{\psi_\lambda}$ by using the oracle for initial state and $U_{P_\lambda}$ $\Or(1)$ times.

We  remark that the eigenstate filtering procedure can also be  implemented by alternative methods such as LCU. The polynomial $R_{\ell}(\cdot,\wt{\Delta})$ can be expanded exactly into a linear combination of the first $2\ell+1$ Chebyshev polynomials. The 1-norm of the expansion coefficients is upper bounded by $2 \ell+2$ 
because $|R_{\ell}(x,\wt{\Delta})|\leq 1$. However, this comes at the expense of additional $\Or(\log \ell)$ qubits needed for the LCU expansion \cite{ChildsKothariSomma2017}.


\REV{Besides the projection operator, we can use this filtering procedure to implement many other related operators. First we consider implementing the reflection operator about the target $\lambda$-eigenstate (or $\lambda$-eigenspace if there is degeneracy), $2P_\lambda-I$, which is useful in the amplitude amplification procedure \cite{grover1996fast,BrassardHoyerMoscaEtAl2002}. This problem has been considered in Ref.~\cite{chowdhury2018improved}.}

\REV{For a given Hamiltonian $H$, with the same assumptions as in Theorem~\ref{thm:eigenstate_filter}, and $\wt{H}=(H-\lambda I)/(\alpha+|\lambda|)$ as constructed above, we define
\[
R_\lambda = 2P_\lambda-I,
\]
where $P_\lambda$ is the projection operator into the $\lambda$-eigenspace of $H$. Using a polynomial $S_\ell(x;\Delta)$ constructed from $R_\ell(x;\Delta)$ as introduced in Appendix~\ref{sec:apdx_ref_and_theta_ref}, we can implement the reflection operator $R_\lambda$ through QSP. The cost is summarized as follows:
\begin{thm}
\label{thm:reflection}
Under the same assumption as Theorem~\ref{thm:eigenstate_filter}, a $(1,m+2,\epsilon)$-block-encoding of $R_\lambda$, the reflection operator about the $\lambda$-eigenspace of $H$, can be constructed using $\mathcal{O}((\alpha/\Delta)\log(1/\epsilon))$ applications of (controlled-) $U_H$ and $U^{\dagger}_H$, and $\mathcal{O}((m\alpha/\Delta)\log(1/\epsilon))$ other primitive quantum gates.
\end{thm}
For the proof see Appendix~\ref{sec:apdx_ref_and_theta_ref}. This reflection operator further enables us to construct a block-encoding of the $\theta$-reflection operator.
\[
P_\lambda + e^{\I \theta}(I-P_\lambda).
\]
This operator is useful in fixed-point amplitude amplification \cite{yoder2014fixed,grover2005fixed}. The cost is summarized as follows:
\begin{cor}
\label{cor:theta_reflection}
Under the same assumption as Theorem~\ref{thm:eigenstate_filter}, a $(1,m+3,\epsilon)$-block-encoding of $P_\lambda+e^{\I\theta}(I-P_\lambda)$, where $P_\lambda$ is the projection operator into the $\lambda$-eigenspace of $H$, can be constructed using $\mathcal{O}((\alpha/\Delta)\log(1/\epsilon))$ applications of (controlled-) $U_H$ and $U^{\dagger}_H$, and $\mathcal{O}((m\alpha/\Delta)\log(1/\epsilon))$ other primitive quantum gates.
\end{cor}
The proof can be found in Appendix~\ref{sec:apdx_ref_and_theta_ref}.
}

\REV{In this paper we focus on obtaining the eigenstate corresponding to an eigenvalue that is known exactly. If instead of a single known eigenvalue, we want keep all eigenvalues in a certain interval, and filter out the rest, we can use a linear combination of polynomials used to approximate the sign function \cite[Lemma~14]{GilyenSuLowEtAl2019}, together with constant shift. The filtering polynomial for this kind of task can also be obtained numerically through Remez algorithm \cite{remez1934determination}, followed by a optimization based procedure to efficiently identify the phase factors. For more details we refer readers to Ref. \cite{DongMengWhaleyLin2020}.}

\REV{
\begin{rem}
In the special case where $\|H\|=1$, the target eigenvalue is 1, and we have access to a $(1,m,0)$-block-encoding of $H$, then a quadratically improved dependence on the gap can be achieved using polynomials such as \cite[Eq.~(6.113)]{Saad2003}. 
This is useful for obtaining the stationary distribution of an ergodic and reversible Markov chain because the discriminant matrix \cite{szegedy2004quantum, apers2019quantum,apers2019unified} can be block-encoded efficiently in a reflection operator, and its $1$-eigenstate is $\sum_j\sqrt{\pi_j}\ket{j}$ where $\pi=(\pi_j)$ is the stationary distribution. 
\end{rem}
}

\section{Solving QLSP: eigenstate filtering  with adiabatic quantum computing}
\label{sec:qlsp_aqc}

To define QLSP, we assume that a $d$-sparse matrix $A$ can be accessed by  oracles $O_{A,1}$, $O_{A,2}$ as
\begin{equation}
O_{A,1}\vert j,l\rangle  =  \vert j,\nu(j,l)\rangle,\quad
O_{A,2}\vert j,k,z\rangle  =  \vert j,k,A_{jk}\oplus z\rangle, 
\label{eq:oracles_A}
\end{equation}
where $j,k,l,z\in [N]$, and $\nu(j,l)$ is the row index of the $l$-th nonzero element in the $j$-th column. The right hand side vector $\ket{b}$ can be prepared with an oracle $O_B$ as
\begin{equation}
O_B\ket{0} = \ket{b}.
\label{eq:oracles_b}
\end{equation}
This is the same as the assumption used in \cite{ChildsKothariSomma2017, SubasiSommaOrsucci2019}. The oracles can be used to construct a $(d,n+2,0)$-block-encoding of $A$ \cite{BerryChildsKothari2015,ChildsKothariSomma2017}.

\REV{We assume the singular values of $A$ are contained in $[1/\kappa,1]$ for some $\kappa>1$. Therefore $\kappa$ here is an upper bound for the condition number, which is defined as the ratio between the largest and the smallest singular values. It is thus guaranteed that when $A$ is Hermitian its eigenvalues are contained in $\mc{D}_{1/\kappa}=[-1,-1/\kappa]\cup[1/\kappa,1]$.  In Ref.~\cite{ChildsKothariSomma2017} it is assumed that that $\|A\|=1$ and the condition number is exactly $\kappa$ \cite[Problem 1]{ChildsKothariSomma2017}, which is slightly stronger than the assumption we are currently using. }

\begin{rem}
\label{rem:dilation_non_hermitian}
\REV{
We can always assume without loss of generality that $A$ is Hermitian. Because when $A$ is not Hermitian we can solve an extended linear system as described in Appendix~\ref{sec:apdx_matrix_dilation}, where the coefficient matrix is
\[
\begin{pmatrix}
0 & A\\
A^\dagger  & 0
\end{pmatrix}
\]
This is a Hermitian matrix, and when $A$ is $d$-sparse, this matrix is $d$-sparse as well.
If $A$ has singular values $\{\sigma_k\}$, then the dilated Hermitian matrix has real eigenvalues $\{\pm\sigma_k\}$. Therefore the two matrices have the same condition number, and when the singular values of $A$ are contained in $[1/\kappa,1]$ the spectrum of the above dilated matrix is contained in $[-1,-1/\kappa]\cup[1/\kappa,1]$.
}
\end{rem}


We will then apply the method we developed in the last section to QLSP. To do this we need to convert QLSP into an eigenvalue problem. For simplicity we assume $A$ is Hermitian positive-definite. The indefinite case is addressed in Appendix~\ref{sec:apdx_matrix_dilation}, which \REV{uses different Hamiltonians but} only requires minor modifications. We define
\begin{equation}
\label{eq:H1}
H_1=\begin{pmatrix}
0 & AQ_b\\
Q_bA & 0
\end{pmatrix}=\ket{0}\bra{1}\otimes AQ_b + \ket{1}\bra{0}\otimes Q_bA,
\end{equation}
where $Q_b = I-\ket{b}\bra{b}$. \REV{This Hamiltonian has been used in Refs.~\cite{SubasiSommaOrsucci2019,AnLin2019}.} As discussed in Appendix~\ref{sec:apdx_block_encoding}, we can construct a $(d,n+4,0)$-block-encoding of $H_1$, denoted by $U_{H_1}$ by applying $O_B, O_{A,1}, O_{A,2}$ twice.

We may readily verify that the 0-eigenspace, \ie the null space, of $H_1$ is spanned by $\ket{0}\ket{x}=(x,0)^\top$, where $\ket{x}$ is the solution, \ie $A\ket{x}\propto\ket{b}$, and $\ket{1}\ket{b}=(0,b)^\top$\REV{, by considering the null space of $H_1^2$}. The rest of the spectrum is separated from 0 by a gap of $1/\kappa$ \cite{AnLin2019,SubasiSommaOrsucci2019}.
Therefore to apply the eigenstate filtering method, we only need an initial state\footnote{\REV{We will later discuss how to use AQC and QZE to prepare this state. Therefore it is worth pointing out that by initial state here we mean the state on which we apply the eigenstate filtering, rather than the initial state of AQC or QZE.} } with non-vanishing overlap \REV{with the target eigenstate $\ket{0}\ket{x}$} that can be efficiently prepared. We will prepare this initial state using the time-optimal adiabatic quantum computing.

\subsection{Choosing the eigenpath}
\label{sec:quantum_zeno_eigenpath}

To use adiabatic quantum computing we need to first specify the eigenpath we are going to follow.
 We define
\begin{equation}
H_0=\begin{pmatrix}
0 & Q_b\\
Q_b & 0
\end{pmatrix}=\sigma_x\otimes Q_b.
\label{eq:H0}
\end{equation}
and
\[
H(f) = (1-f)H_0 + fH_1,
\]
\REV{where $H_1$ is defined in Eq.~\eqref{eq:H1}.}

We will then evolve the system following the 0-eigenstates of each $H(f)$. These eigenstates form an eigenpath linking the initial state to the solution to the linear system. There are several important properties of the Hamiltonians $H(f)$ and of the eigenpath which we discuss below, though some of them we will only use in the algorithm based on the quantum Zeno effect.

The null space of $H(f)$ is two-dimensional, and we will pay special attention to this fact in our analysis.
The non-zero eigenvalues of $H(f)$ appear in pairs. Let \REV{$\lambda_j(f)$, $j=1,2,\ldots,N-1$} be all the positive eigenvalues of $H(f)$, and $\ket{z_j(f)}$ be the corresponding eigenvectors, then we may readily check
\[
H(f)(\sigma_z\otimes I)\ket{z_j(f)} = -\lambda_j(f)(\sigma_z\otimes I)\ket{z_j(f)}.
\]
Therefore $-\lambda_j(f)$ is also an eigenvalue of $H(f)$ with corresponding eigenvector $(\sigma_z\otimes I)\ket{z_j(f)}$, for \REV{$j=1,2,\ldots,N-1$}. Thus we have obtained all the non-zero eigenvalues and corresponding eigenvectors. 

The form of the matrices in Eqs.~\eqref{eq:H1} and \eqref{eq:H0} is important for achieving $\Or(\kappa)$ complexity in our algorithms because they ensure the gap between 0 and other eigenvalues for all $f$ is lower bounded by
\begin{equation}
\label{eq:gap_aqc}
\Delta_*(f) = 1-f+\frac{f}{\kappa}.
\end{equation}
A proof can be found in \cite{AnLin2019}.

Now we are ready to specify the eigenpath. 
\REV{For any $f$, we let $\ket{x(f)}$ be some vector such that 
\begin{equation}
\label{eq:def_xf}
((1-f)I+fA)\ket{x(f)}\propto \ket{b}.
\end{equation}
}
We can then see that the null space of $H(f)$ is spanned by $\ket{\bar{x}(f)}=\ket{0}\ket{x(f)}$ and $\ket{1}\ket{b}$.
This requirement pins down the choice for $\ket{x(f)}$ up to a time-dependent global phase. 
\REV{By requiring the phase to be geometric, i.e.
\begin{equation}
\label{eq:geometric_phase}
\braket{x(f)|\partial_f|x(f)}=0,
\end{equation}
the eigenpath $\{\ket{x(f)}\}$ becomes uniquely defined when we require $\ket{x(0)}=\ket{b}$. Note the above equation is slightly problematic in that we do not know beforehand that $\ket{x(f)}$ is differentiable. However this turns out not to be a problem because we can establish the differentiability in Appendix~\ref{sec:apdx_eigenpath}.  Furthermore, we have the estimate
\begin{equation}
\label{eq:bound_eigenpath_derivative}
\begin{aligned}
\|\partial_f \ket{x(f)}\| 
&\leq \frac{2}{\Delta_*(f)}.
\end{aligned}
\end{equation}
The derivation of the existence and uniqueness of the differentiable eigenpath, together with the estimate \eqref{eq:bound_eigenpath_derivative} are given in Appendix~\ref{sec:apdx_eigenpath}.
}

An important quantity we need to use in our analysis is the eigenpath length
\[
L=\int_0^1 \|\partial_f\ket{x(f)}\|\text{d}f,
\]
and by (\ref{eq:bound_eigenpath_derivative}) we have
\begin{equation}
\label{eq:bound_eigenpath_length}
L\leq \int_0^1 \frac{2}{\Delta_*(f)}\text{d}f = \frac{2\log(\kappa)}{1-1/\kappa}.
\end{equation}
We also define the eigenpath length $L(a,b)$ between $0<a<b<1$ and it is bounded by
\begin{equation}
\label{eq:bound_segment_eigenpath_length}
L(a,b)=\int_a^b \|\partial_f\ket{x(f)}\|\text{d}f\leq \frac{2}{1-1/\kappa}\log\left(\frac{1-(1-1/\kappa)a}{1-(1-1/\kappa)b}\right)=:L_{*}(a,b).
\end{equation}
%

\subsection{Time-optimal adiabatic quantum computing}

Here we briefly review the procedure of solving QLSP using the recently developed time-optimal AQC~\cite{AnLin2019} and the eigenpath described in the previous section \REV{that has been used in \cite{AnLin2019,SubasiSommaOrsucci2019}}.

As noted before, the null space of $H(f)$ is two-dimensional, which contains an unwanted  0-eigenvector $\ket{1}\ket{b}=(0,b)^\top$. 
However this 0-eigenvector is not accessible in the AQC time-evolution
$$
\frac{1}{T}\I \partial_s \left|\psi_T(s)\right> = H(f(s))\left|\psi_T(s)\right>, \quad    \ket{\psi_T(0)}=\ket{0}\ket{b},
$$
for scheduling function $f:[0,1]\rightarrow [0,1]$, which is a strictly increasing mapping   with $f(0) = 0, f(1) = 1$. We find  that  
\[
(\bra{1}\REV{\bra{b}}) \ket{\psi_T(s)} = 0,
\]
\REV{
for all $s\in[0,1]$. This is due to 
$$
\frac{1}{T}\I \partial_s (\bra{1}\REV{\bra{b}})\left|\psi_T(s)\right> =(\bra{1}\REV{\bra{b}}) H(f(s))\left|\psi_T(s)\right>=0,
$$
and $(\bra{1}\REV{\bra{b}}) \ket{\psi_T(0)} = 0$.
}
This fact gets rid of the problem.

The parameter $T$ needed to reach a certain target accuracy $\epsilon$ is called the runtime complexity (or simply the time complexity).
The simplest choice for the scheduling function is $f(s)=s$, which gives the ``vanilla AQC''. Besides $\ket{0}\ket{x}$, all other eigenstates of $H_1$ that can be connected to $\ket{0}\ket{b}$ through an adiabatic evolution are separated from $\ket{0}\ket{x}$ by an energy gap of at least $1/\kappa$ \cite{AnLin2019,SubasiSommaOrsucci2019}. The time complexity of vanilla AQC is at least $T\sim \Or(\kappa^2/\epsilon)$  \cite{JansenRuskaiSeiler2007,AnLin2019,AlbashLidar2018,ElgartHagedorn2012}.

By properly choosing a scheduling function $f(s)$, the time complexity of AQC can be significantly improved. There are two time-optimal scheduling functions proposed in \cite{AnLin2019}. The first method is called AQC(p). For $1<p<2$, AQC(p) adopts the schedule
\begin{equation}
    f(s) = \frac{\kappa}{\kappa - 1}\left[1-\left(1+s(\kappa^{p-1}-1)\right)^{\frac{1}{1-p}}\right]{.}
    \label{eqn:AQCSchedule_explicit}
\end{equation}
This reduces the time complexity to $\Or(\kappa/\epsilon)$, which is optimal for $\kappa$, but the scaling with respect to $\epsilon$ is the same.  The second method is called AQC(exp), which uses a different scheduling function to achieve time complexity
$\Or\left(\kappa\log^2(\kappa)\log^4\left(\frac{\log\kappa}{\epsilon}\right)\right)$. 

All AQC methods are time-dependent Hamiltonian simulation problem, which can be implemented using e.g. truncated Dyson series for simulating the time-dependent Hamiltonian \cite{LowWiebe2018}. Although AQC(exp) scales near-optimally with respect to $\kappa$ and $\epsilon$, numerical evidence indicates that the preconstant of AQC(exp) can be higher than AQC(p). Hence when a low accuracy $\epsilon\sim \Or(1)$ is needed, AQC(p) can require a smaller runtime in practice. In the discussion below, we will consider AQC(p).

The details of the time-dependent Hamiltonian simulation for AQC are discussed in Appendix~\ref{sec:apdx_aqc}, and
the query complexity for implementing AQC(p) on a gate-based quantum computer is $\wt{\Or}(\kappa/\epsilon)$.  

\subsection{Improved dependence on $\epsilon$}

We now use eigenstate filtering to accelerate AQC(p) and reduce the query complexity to $\log(1/\epsilon)$. As mentioned before, once we have access to $H_1$ defined in \eqref{eq:H1}, through the block-encoding $U_{H_1}$ constructed in Appendix~\ref{sec:apdx_block_encoding} we only need an initial state \REV{for eigenstate filtering (note that this is not the initial state of the AQC time-evolution)}:
\begin{equation}
\ket{\wt{x}_0} = \gamma_0 \ket{0}\ket{x} + \gamma_1 \ket{1}\ket{b} + \ket{\perp} 
\label{eqn:x0expand}
\end{equation}
with $|\gamma_0|=\Omega(1)$ and $\ket{\perp}$ orthogonal to the null space. The initial state $\ket{\wt{x}_0}$ can be prepared using the time-optimal AQC procedure. Again we first assume $A$ is Hermitian positive definite. To make $|\gamma_0|=\Omega(1)$ we only need to run AQC(p) to constant precision, and thus the linear dependence on precision is no longer a problem. Therefore the time complexity of AQC(p) is $\mathcal{O}(\kappa)$. However we still need to implement AQC(p) on a quantum circuit. To do this we use the time-dependent Hamiltonian simulation introduced in \cite{LowWiebe2018}, which gives a $\Or(d\kappa\log(d\kappa)/\log\log(d\kappa))$ query complexity to achieve $\Or(1)$ precision\REV{, for a $d$-sparse matrix $A$}. This procedure also needs to be repeated $\mathcal{O}(1)$ times. It should be noted that $\gamma_1$ in Eq.~\eqref{eqn:x0expand} comes entirely from the error of the Hamiltonian simulation, since AQC should ensure that the state is orthogonal to $\ket{1}\ket{b}$ for all $t$. Details on performing this time-dependent Hamiltonian simulation is given in Appendix~\ref{sec:apdx_aqc}.

Then we can run the eigenstate filtering algorithm described in Section~\ref{sec:eigstate_filtering} to precision $\epsilon$ to obtain $R_{\ell}(H_1/d;1/(d\kappa))\ket{\wt{x}_0}$. The $\ket{\perp}$ component will be filtered out, while the $\ket{0}\ket{x}$ and $\ket{1}\ket{b}$ components remain. To further remove the $\ket{1}\ket{b}$ component, we measure the first qubit. Upon getting an outcome 0, the outcome state will just be $\ket{0}\ket{x}+\mathcal{O}(\epsilon)$. The success probability of applying the eigenstate filtering is lower bounded by $|\gamma_0|^2+|\gamma_1|^2$, and the success probability of obtaining 0 in measurement is $|\gamma_0|^2/(|\gamma_0|^2+|\gamma_1|^2)+\mathcal{O}(\epsilon)$. Thus the total success probability is $\Omega(1)$. Each single application of eigenstate filtering applies $U_{H_1}$, and therefore $O_{A,1}$, $O_{A,2}$, and $O_{B}$, for $\Or(d\kappa\log(1/\epsilon))$ times. It only needs to be repeated $\Omega(1)$ times so the total query complexity of eigenstate filtering is still $\Or(d\kappa\log(1/\epsilon))$.

In eigenstate filtering we need $\Or(nd\kappa\log(1/\epsilon))$ additional primitive gates as mentioned in Theorem~\ref{thm:eigenstate_filter}. In time-dependent Hamiltonian simulation the addition number of primitive gates needed is $\Or(d\kappa(n+\log(d\kappa))\frac{\log(d\kappa)}{\log\log(d\kappa)})$. Both procedures are repeated $\Or(1)$ times and therefore in total we need 
$
\Or\left(d\kappa\left(n\log(\frac{1}{\epsilon})+(n+\log(d\kappa))\frac{\log(d\kappa)}{\log\log(d\kappa)}\right)\right)
$
additional primitive gates.

The number of qubits needed in the eigenstate filtering procedure using QSP is $\mathcal{O}(n)$ which mainly comes from the original size of the problem and block-encoding. Extra ancilla qubits introduced as a result of eigenstate filtering is only $\Or(1)$. 
In the Hamiltonian simulation $\mathcal{O}(n+\log(d\kappa))$ qubits are needed (see \REV{Appendix~\ref{sec:apdx_aqc}}). Therefore the total number of qubits needed is $\mathcal{O}(n+\log(d\kappa))$.

The procedure above can be generalized to Hermitian indefinite matrices, and general matrices \REV{that are not necessarily Hermitian (see Appendix~\ref{sec:apdx_matrix_dilation}).  As discussed in Remark~\ref{rem:dilation_non_hermitian}, for general matrices we should assume the singular values instead of eigenvalues of $A$ are contained in $[1/\kappa,1]$.}
Therefore our QLSP solver can be summarized as

\begin{thm}
\label{thm:qlsp}
 $A$ is a $d$-sparse matrix whose \REV{singular values are in $[1/\kappa,1]$} and can be queried through oracles $O_{A,1}$ and $O_{A,2}$ in \eqref{eq:oracles_A}, and $\ket{b}$ is given by an oracle $O_B$ in \eqref{eq:oracles_b}. Then  $\ket{x}\propto A^{-1}\ket{b}$ can be obtained with fidelity 
$1-\epsilon$, \REV{succeeding with probability $\Omega(1)$ with ancilla qubits measurement outcome indicating success,} using
 
 1. $\mathcal{O}\left(d\kappa(\frac{\log(d\kappa)}{\log\log(d\kappa)}+\log(\frac{1}{\epsilon}))\right)$ queries to $O_{A,1}$, $O_{A,2}$, and $O_B$,
 
 2. $
\Or\left(d\kappa\left(n\log(\frac{1}{\epsilon})+(n+\log(d\kappa))\frac{\log(d\kappa)}{\log\log(d\kappa)}\right)\right)
$ other primitive gates,

 3. $\mathcal{O}(n+\log(d\kappa))$ qubits.
\end{thm}
 When the gate complexity of $O_{A,1}$, $O_{A,2}$, and $O_B$ are $\text{poly}(n)$ the total gate complexity, and therefore runtime, by the above theorem, will be $\wt{\Or}(\text{poly}(n)d\kappa\log(1/\epsilon))$.

\begin{rem}
\REV{Although in total we need $\mathcal{O}(n+\log(d\kappa))$ ancilla qubits, only $\mathcal{O}(\log(d\kappa))$ comes sources other than the block-encoding of $A$. In other words, our method only adds $\mathcal{O}(\log(d\kappa))$ ancilla qubits to those that are unavoidable as long as we use this way of block-encoding of a sparse $A$.  These extra ancilla qubits are mainly a result of using time-dependent Hamiltonian simulation.  Also, although in the theorem we assumed $A$ is a sparse matrix, we have only used this fact to build its block-encoding. Given the block-encoding of a matrix $A$ that is not necessarily sparse, the above procedure can still be carried out directly. This is also true for Theorem~\ref{thm:qlsp_zeno} which we are going to introduce later.}
\end{rem}

We present numerical results \REV{obtained on a classical computer} in Fig.~\ref{fig:qls_results} to validate the complexity estimate. In the numerical test, \REV{we solve the linear system $A\ket{x}\propto \ket{b}$, where} $A$ is formed by adding a randomly generated symmetric positive definite tridiagonal matrix $B$, whose smallest eigenvalue is very close to 0, to a scalar multiple of the identity matrix. After properly rescaling, the eigenvalues of $A$ lie in $[-1,1]$. This construction enables us to estimate condition number with reasonable accuracy without computing eigenvalues. 
The off-diagonal elements of $B$ are drawn uniformly from $[-1,0]$ and the diagonal elements are the negative of sums of two adjacent elements on the same row. The $(0,0)$ and $(N-1,N-1)$ elements of $B$ are slightly larger so that $B$ is positive definite. \REV{$\ket{b}$ is drawn from the uniform distribution on the unit sphere.} 

\REV{With $A$ and $\ket{b}$ chosen, we first run the AQC time evolution for time $\Or(\kappa)$ as described at the beginning of this section, and then apply eigenstate filtering using the polynomial $R_\ell(x;1/d\kappa)$ with degree $2\ell$. Denoting the resulting quantum state by $\ket{\wt{x}}$ we then compute the fidelity $\eta=\abs{\braket{x|\wt{x}}}$. Fig.~\ref{fig:qls_results} shows the relation between $\eta$, $\kappa$, and $\ell$ obtained in the numerical experiment.}

\begin{figure}[ht!]
\begin{center}
\includegraphics[width=0.8\textwidth]{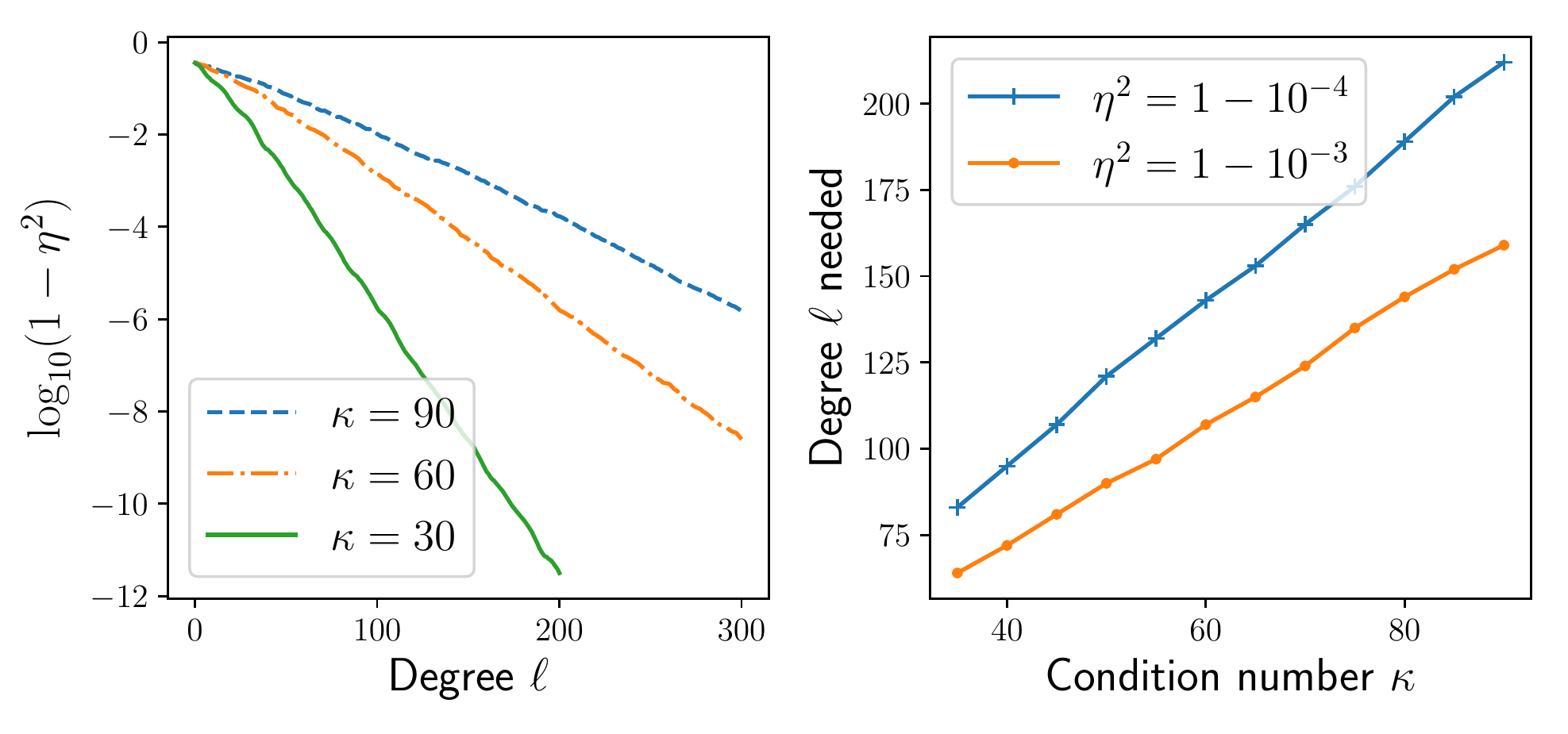}
\end{center}
\caption{Left: fidelity $\eta$ converges to 1 exponentially as $\ell$ in the eigenvalues filtering algorithm increases, for different $\kappa$. Right: the smallest $\ell$ needed to achieve fixed fidelity $\eta$ grows linearly with respect to condition number $\kappa$. The initial state in eigenstate filtering is prepared by running AQC(p) for $T=0.2\kappa$, with $p=1.5$, which achieves an initial fidelity of about 0.6.}
\label{fig:qls_results}
\end{figure}

\section{Solving QLSP: eigenstate filtering  with quantum Zeno effect}
\label{sec:quantum_zeno}

Quantum Zeno effect (QZE) \REV{is the phenomenon that frequent measurements hinders a quantum system's transition from its initial state to other states \cite{misra1977zeno,facchi2008quantum,facchi1999berry,balachandran2000quantum,burgarth2013non}.} A variant of QZE \cite[Lemma 1]{BoixoKnillSomma2009} can be viewed as a particular way for implementing adiabatic quantum computing \cite{poulin2018quantum,somma2008quantum,lemieux2020resource}, \REV{and this is what we mean by QZE throughout this work unless stated otherwise.}  The basic idea of \REV{this variant of} QZE is to follow an adiabatic path through repeated measurement, which acts as projection operators to the instantaneous eigenstate along the adiabatic path. This inspired the randomization method for performing computation based on QZE \cite{BoixoKnillSomma2009,SubasiSommaOrsucci2019}. 

In the context of solving QLSP,  again for simplicity we first assume $A$ is Hermitian positive definite. Instead of running time-dependent Hamiltonian simulation to evolve from the 0-eigenstate  of $H_0$ to the 0-eigenstate of $H_1$, we consider applying a series of projections to traverse the eigenpath. Choosing $0=f_0<f_1<\ldots<f_{M}=1$, for each $j=0,1,\ldots,M-1$, we start from the 0-eigenstate \REV{$\ket{0}\ket{x(f_j)}$} of \REV{$H(f_j)$,} \REV{where $\ket{x(f)}$ is defined in Eqs.~\eqref{eq:def_xf} and \eqref{eq:geometric_phase},} and project into the null space of $H(f_{j+1})$. In the end we obtain the 0-eigenstate of $H(1)=H_1$. This is essentially the same as performing projective measurement for each $j$ \cite{ChildsDeottoFarhiEtAl2002,BoixoKnillSomma2009,SubasiSommaOrsucci2019}. \REV{If} the projective measurements are done approximately using quantum phase estimation or phase randomization, \REV{there will be} a linear  dependence on $1/\epsilon$ in runtime, $\epsilon$ being the desired precision.

In this section we combine eigenstate filtering with Zeno-based computation to reduce the error dependence from $\Or(1/\epsilon)$ to $\Or(\log(1/\epsilon))$, thanks to the possibility of performing approximate projections with high precision. However, several issues demand our attention in the procedure outlined at the beginning of this section. First, we need to specify the choice of $\{f_j\}$, which plays an important role in the lower bound of $M$ needed to ensure at least constant success probability. Second, the null space of each $H(f_j)$ is 2-dimensional. Therefore the eigenpath is not unique, and we need to specify the eigenpath we are going to traverse, which has been done in Sec.~\ref{sec:quantum_zeno_eigenpath}, and to ensure the undesired part of the null space does not interfere with our computation.

\subsection{The algorithm}
\label{sec:quantum_zeno_alg}

\REV{As in Section~\ref{sec:qlsp_aqc}, the goal is to produce a state close to the solution state $\ket{x}$ of the QLSP with fidelity at least $1-\epsilon$ for some given $0<\epsilon<1$.}
In this section we describe the procedure of the Quantum Zeno effect state preparation. We need to choose a scheduling function
\begin{equation}
\label{eq:scheduling}
f(s)=\frac{1-\kappa^{-s}}{1-\kappa^{-1}}
\end{equation}
and define $f_j=f(s_j)$ where $s_j=j/M$. \REV{Without the scheduling we will end up with an unfavorable square dependence on the minimum spectral gap along the eigenpath \cite{ChildsDeottoFarhiEtAl2002}.} This scheduling is chosen so that
\begin{equation}
\label{eq:curve_length_reparam}
L(f_j,f_{j+1})\leq L_*(f_j,f_{j+1})=\frac{2\log(\kappa)}{M(1-1/\kappa)},
\end{equation}
which implies we are dividing \REV{the interval $[0,1]$ of $f$} into $M$ segments of equal $L_*$-length.

Before we describe the algorithm we need to first introduce some notations and block-encodings we need to use. From the block-encoding of $H_0$ and $H_1$ described in Appendix~\ref{sec:apdx_block_encoding}, we can construct $(1-f+fd,n+6,0)$-block-encoding for each $H(f)$, denoted by $U_H(f)$. This construction uses  \cite[Lemma 29]{GilyenSuLowEtAl2019}\REV{, through
\[
H(f) = (1-f+fd)\left(\bra{c}\otimes I\right)\left[ \ket{0}\bra{0}\otimes H_0 + \ket{1}\bra{1}\otimes (H_1/d) \right]\left(\ket{c}\otimes I\right),
\]
where
\[
\ket{c}=\frac{1}{\sqrt{1-f+fd}}(\sqrt{1-f}\ket{0}+\sqrt{fd}\ket{1}).
\]
We need to use $H_1/d$ instead of $H_1$ because there is a $d$ factor involved in the block-encoding of $H_1$ (see Appendix~\ref{sec:apdx_block_encoding}) , and the above equation shows we get a $1-f+fd$ factor in the block-encoding of $H(f)$ because we need to normalize the coefficient vector $\ket{c}$. For a more detailed discussion see Appendix~\ref{sec:apdx_block_encoding}.}
Applying the eigenstate filtering procedure in Section~\ref{sec:eigstate_filtering} to precision $\epsilon_P$ gives us an $(1,n+7,\epsilon_P)$-block-encoding of 
\begin{equation}
\label{eq:def_P0_bar}
\bar{P}_0(f)=\ket{0}\ket{x(f)}\bra{0}\bra{x(f)} + \ket{1}\ket{b}\bra{1}\bra{b},
\end{equation}
which we denote by $U_{P_0}(f)$. By Theorem~\ref{thm:eigenstate_filter} this uses $U_{H}(f)$ and its inverse $\Or(\frac{d}{\Delta_*(f)}\log(\frac{1}{\epsilon_P}))$ times. Note that one ancilla qubit introduced in Theorem~\ref{thm:eigenstate_filter} is redundant because we do not need to shift by a multiple of the identity matrix. By definition of block-encoding we have
\[
\Big\| \bar{P}_0(f) - (\bra{0^{n+7}}\otimes I_{n+1})U_{P_0}(f)(\ket{0^{n+7}}\otimes I_{n+1}) \Big\| \leq \epsilon_P.
\]
Here for clarity we use $I_r$ to denote the identity operator acting on $r$ qubits. Note that we need access to
\begin{equation}
\label{eq:exact_projector}
P_0(f) = \ket{x(f)}\bra{x(f)},
\end{equation} 
\REV{which is the projection operator onto $\ket{x(f)}$,} instead of $\bar{P}_0(f)$, \REV{which is the projection operator onto $\ket{0}\ket{x(f)}$. We now consider how to approximate $P_0(f)$.} Because of the fact 
\[
P_0(f) = (\bra{0}\otimes I_n)\bar{P}_0(f)(\ket{0}\otimes I_n),
\]
we denote 
\begin{equation}
\label{eq:def_wtP0}
\wt{P}_0(f)=(\bra{0^{n+7}}\bra{0}\otimes I_n)U_{P_0}(f)(\ket{0^{n+7}}\ket{0}\otimes I_n)
\end{equation}
and $\wt{P}_0(f)$ approximates $P_0(f)$ by the following inequalities:
\[
\begin{aligned}
\|\wt{P}_0(f)-P_0(f)\| &=\Big\| (\bra{0}\otimes I_n) \left( (\bra{0^{n+7}}\otimes I_1\otimes I_n)U_{P_0}(f)(\ket{0^{n+7}}\otimes I_1\otimes I_n) - \bar{P}_0(f) \right) (\ket{0}\otimes I_n) \Big\| \\
&\leq \Big\|  (\bra{0^{n+7}}\otimes I_1\otimes I_n)U_{P_0}(f)(\ket{0^{n+7}}\otimes I_1\otimes I_n) - \bar{P}_0(f) \Big\| \\
&\leq \epsilon_P.
\end{aligned}
\]
Therefore $U_{P_0}(f)$ is an $(1,n+8,\epsilon_P)$-block-encoding of $P_0(f)$.

\REV{As discussed in Section~\ref{sec:quantum_zeno_eigenpath}, the eigenpath we want to follow is $\{\ket{0}\ket{x(f)}\}$. 
However the approximate projection using eigenstate filtering only allows us to approximately follow this eigenpath. We denote the approximate states by $\ket{\wt{x}(f_j)}\approx \ket{x(f_j)}$, and will take into account the error of this approximation in our analysis.}


With the block-encoding of $P_0(f)$ we can describe the algorithm is as follows:
\begin{enumerate}
\item \REV{Given $0<\epsilon<1$ and $\kappa>1$ as well as the oracles mentioned at the beginning of Section~\ref{sec:qlsp_aqc}.} Set $M=\REV{\lceil \frac{4\log^2(\kappa)}{(1-1/\kappa)^2} \rceil}$, $\epsilon_P=\frac{1}{162M^2}$.
\item Prepare $\ket{\wt{x}(0)}=\ket{b}$. Let $j=1$.
\item Apply the $(1,n+8, \epsilon_P)$-block-encoding $U_{P_0}(f_j)$ of $P_0(f_j)$\REV{, constructed using eigenstate filtering with a polynomial of sufficiently high degree constructed in Lemma~\ref{lem:minimax_poly},} to $\ket{0^{n+8}}\ket{\wt{x}(f_{j-1})}$ to get $U_{P_0}(f_j)(\ket{0^{n+8}}\ket{\wt{x}(f_{j-1})})$.
\item Measure the $n+8$ ancilla qubits. 
\begin{itemize}
 \item[\REV{(a)}] If not all outputs are 0 then abort and return to Step 2. 
 \item[\REV{(b)}] If all outputs are 0, and further $j< M-1$, then let $\ket{\wt{x}(f_j)}$ be the state in the main register that has not been measured, let $j\leftarrow j+1$, and go to Step 3. If all outputs are 0 and $j=M-1$ then go to next step.
\end{itemize}
\item Apply the \REV{$(1,n+8, \epsilon/4)$}-block-encoding $U_{P_0}(1)$ of $P_0(1)$ to $\ket{0^{n+8}}\ket{\wt{x}(f_{M-1})}$ to get $U_{P_0}(f_j)(\ket{0^{n+8}}\ket{\wt{x}(f_{M-1})})$.
\item Measure the $n+8$ ancilla qubits. 
\begin{itemize}
    \item[\REV{(a)}] If not all outputs are 0 then abort and return to Step 2. 
    \item[\REV{(b)}] If all outputs are 0, then output $\ket{\wt{x}(1)}$ in the main register.
\end{itemize}
\end{enumerate}

\REV{Here $\ket{\wt{x}(f_j)}$ are defined recursively in Steps 3 and 4 in the algorithm, starting with $\ket{\wt{x}(0)}=\ket{b}$. We can write down the recursion more concisely:
\begin{equation}
\label{eq:approx_xf}
\ket{\wt{x}(f_{j})}=\frac{\wt{P}_0(f_j)\ket{\wt{x}(f_{j-1})}}{\|\wt{P}_0(f_j)\ket{\wt{x}(f_{j-1})}\|}.
\end{equation}
Going from $\ket{\wt{x}(f_{j-1})}$ to $\ket{\wt{x}(f_j)}$ has a success probability $\|\wt{P}_0(f_j)\ket{\wt{x}(f_{j-1})}\|^2$. We will show in the next section as well as in Appendix~\ref{sec:apdx_success_probability} that the the final success probability, which is the product of the success probabilities of these individual steps, does not go to 0.  We emphasize that $\{\ket{\wt{x}(f)}\}$ is defined only for $f=f_j$ rather than arbitrary $f\in [0,1]$. We use this notation only to be consistent with the notation $\ket{x(f)}$.}

\begin{rem}[Choice of precision parameters]
\REV{There are two precision parameters involved in the above discussion: $\epsilon$ and $\epsilon_P$. Here $\epsilon$ is the target accuracy specified as part of our task, while $\epsilon_P$ is a parameter that is chosen by the algorithm according to Step 1, and is used only to ensure that the success probability is lower bounded by a constant. Also note that in the last step with $j=M$ (Steps 5 and 6), we set the target accuracy to be \REV{$\epsilon/4$} instead of $\epsilon_P$ in the previous steps. In fact, the errors of eigenstate filtering for $j=1,2,\ldots,M-1$ do not directly contribute to the final error. Rather, they only directly affect the success probability. 
When the overlap $|\braket{\wt{x}(f_{M-1})|x(1)}|$ is lower bounded by a constant away from 0, as we will show in Lemma~\ref{lem:bound_overlap}, the final error is entirely controlled by the accuracy of the final eigenstate filtering for $j=M$, which is in turn controlled by the parameter $\epsilon/4$.}
\REV{In this way we ensure, as will be shown in the next section, that }the output $\ket{\wt{x}(1)}$ satisfies
\[
|\braket{\wt{x}(1)|x}| \geq 1-\epsilon.
\]
\end{rem}


\subsection{Success probability, fidelity, and complexities}
\label{sec:zeno_qlsp_estimates}

\REV{In this section we discuss the success probability of the algorithm described in the previous section, prove the fidelity of the output state is lower bounded by $1-\epsilon$ for the given $\epsilon$ when $\epsilon_P$ and $M$ are chosen as in Step 1 of the algorithm, and finally estimate the query and gate complexities. }

We first give a lower bound for success probability assuming \REV{for simplicity} each projection is done without error, \ie $\epsilon_P=0$. \REV{This is done so that we do not need to distinguish between eigenstates and approximate eigenstates produced using eigenstate filtering, thus making the derivation less technical.} A rigorous lower bound\REV{, assuming a finite $\epsilon_P>0$,} will be given in Appendix~\ref{sec:apdx_success_probability}. Under this assumption we have
\[
\psucc = \prod_{j=1}^{M}\|P_0(f_j)\ket{x(f_{j-1})}\|^2 = \prod_{j=1}^{M}|\braket{x(f_j)|x(f_{j-1})}|^2.
\]
Since
\begin{align}
& |\braket{x(f_j)|x(f_{j-1})}| \geq 1-\frac{1}{2}\|\ket{x(f_{j-1})}-\ket{x(f_j)}\|^2, \label{eq:xf_overlap}\\ 
& \|\ket{x(f_{j-1})}-\ket{x(f_j)}\| \leq L(f_{j-1},f_{j}) \leq L_*(f_{j-1},f_j), \label{eq:xf_distance}
\end{align}
we have
\[
\begin{aligned}
\psucc &\geq \left( \prod_{j=1}^{M} \left(1-\frac{1}{2}\|\ket{x(f_{j-1})}-\ket{x(f_j)}\|^2\right) \right)^2 \\
&\geq \left(1-\frac{2\log^2(\kappa)}{M^2(1-1/\kappa)^2}\right)^{2M} \\
&\geq \left( 1-\frac{2\log^2(\kappa)}{M(1-1/\kappa)^2} \right)^2 \\
&\geq \frac{1}{4},
\end{aligned}
\]
where the we have used Eq.~\eqref{eq:curve_length_reparam}. This inequality holds for $M\geq \frac{4\log^2(\kappa)}{(1-1/\kappa)^2}$ as required in the previous section.

Therefore we have shown the success probability is lower bounded by $1/4$. The success probability when taking into account errors in each approximate projection\REV{, or in other words when we choose $\epsilon_P=1/162M^2$ according to our algorithm rather than setting it to 0,} is still lower bounded by a constant, which is proved in Appendix~\ref{sec:apdx_success_probability}. 

\REV{We then analyze the fidelity and complexities of our algorithm. Here we no longer assume $\epsilon_P=0$, and the following discussion is therefore rigorous.}  In Appendix~\ref{sec:apdx_success_probability} it is shown that 
\[
|\braket{\wt{x}(f_j)|x(f_{j+1})}|\geq 1-\frac{1}{2M}-4\epsilon_P-2\sqrt{2\epsilon_P} \geq 
\frac{1}{2},\quad j=0,1,\ldots,M-1,
\]
for $\epsilon_P\leq 1/128$ and $M\REV{\geq\frac{4\log^2(\kappa)}{(1-1/\kappa)^2}}\geq 4$. \REV{Therefore $|\braket{\wt{x}(f_{M-1})|x(f_M)}|\geq 1/2$,} which allows us to bound the error as,
\REV{\begin{equation}
\begin{aligned}
|\braket{x|\wt{x}(1)}|  
&= |\braket{\wt{x}(f_{M})|x(f_{M})}|  \\ 
&= \frac{|\braket{\wt{x}(f_{M-1})|\wt{P}_0(f_{M})|x(f_{M})}|}{\|\wt{P}_0(f_{M})\ket{\wt{x}(f_{M-1})}\|} \\
&\geq  \frac{|\braket{\wt{x}(f_{M-1})|{P}_0(f_{M})|x(f_{M})}|-\epsilon/4}{\|{P}_0(f_{M})\ket{\wt{x}(f_{M-1})}\|+\epsilon/4} \\
&= \frac{|\braket{\wt{x}(f_{M-1})|x(f_{M})}|-\epsilon/4}{|\braket{\wt{x}(f_{M-1})|x(f_{M})}|+\epsilon/4}
 \\
&\geq 1-\frac{\epsilon/2}{|\braket{\wt{x}(f_{M-1})|x(f_{M})}|} \\
&\geq 1-\epsilon.
\end{aligned}
\end{equation}
The derivation is similar to that of Eq.~\eqref{eq:fidelity_time_j}, and we have used the fact that $\|\wt{P}_0(f_M)-P_0(f_M)\|\leq \epsilon/4$ because in Step 5 our algorithm in the previous section sets the eigenstate filtering accuracy to be $\epsilon/4$ instead of $\epsilon_P$.}
Therefore the state $\ket{\wt{x}(1)}$ prepared in this way has a fidelity at least $1-\epsilon$.

We then estimate the computational costs. At each $j$ we need to apply an $(1,n+8,\epsilon_P)$-block-encoding \REV{$U_{P_0}(f_j)$ of $P_0(f_j)$} to $\ket{\wt{x}(f_{j-1})}$ obtained form the last step. From the analysis in Appendix~\ref{sec:apdx_success_probability} we need \REV{$\epsilon_P\leq 1/162 M^2$}. Therefore we need to apply $U_H(f_{j})$ and its inverse $\Or\left(\frac{1-f_j+df_j}{\Delta_*(f_j)}\log(\frac{1}{\epsilon_P})\right)$ times. In total for $j=1,2,\ldots,M-1$ the number of queries to $U_H(f)$ is of the order
\begin{equation}
\label{eq:complexity_first_M_minus_1}
\begin{aligned}
\log\left(\frac{1}{\epsilon_P}\right)\sum_{j=1}^{M-1}\frac{1-f(s_j)+f(s_j)d}{1-f(s_j)+f(s_j)/\kappa} &\leq \log\left(\frac{1}{\epsilon_P}\right)M\int_0^1 \frac{1-f(s)+f(s)d}{1-f(s)+f(s)/\kappa}\text{d}s \\
&= \log\left(\frac{1}{\epsilon_P}\right)M\left( \frac{d\kappa-1}{\log(\kappa)} \REV{-  \frac{d-1}{1-1/\kappa}} \right),
\end{aligned}
\end{equation}
\REV{for a $d$-sparse matrix $A$ and $\kappa$ is the condition number of $A$.}
Then in the last step for $j=M$\REV{, which is Step~5 in the algorithm in Section~\ref{sec:quantum_zeno_alg},} we need to achieve accuracy \REV{$\epsilon/4$ for the eigenstate filtering}. Therefore we need to apply the block-encoding $U_{P_0}(1)$ with $\Or(d\kappa\log(\frac{1}{\epsilon}))$ queries to $U_H(1)$. As $M=\Or(\log^2(\kappa))$, adding the query complexity of the last step to (\ref{eq:complexity_first_M_minus_1}), and using the fact $\epsilon_P=\Or(1/M^2)$, gives us the total query complexity of a single run  
\begin{equation}
\label{eq:complexity_single_run}
\Or\left(d\kappa\left(\log(\kappa)\log\log(\kappa)+\log(1/\epsilon)\right)\right).
\end{equation}
Because the success probability is $\Omega(1)$, the procedure needs to be run for an expected $\Or(1)$ times to be successful, and therefore the total complexity remains the same. Since $U_H(f)$ queries $O_{A,1}$, $O_{A,2}$, and $O_B$ each $\Or(1)$ times, Eq. (\ref{eq:complexity_single_run}) is also the query complexity to these oracles.

Because the only thing we need to do in this method to solve QLSP is to repeatedly use QSP to do projection, no additional qubits are involved for time-dependent Hamiltonian simulation as in the previous AQC-based method. The total number of qubits is therefore $\Or(n)$. The number of additional primitive gates required can be estimated similarly to the number of queries, which scales as
$
\Or\left(nd\kappa\left(\log(\kappa)\log\log(\kappa)+\log(\frac{1}{\epsilon})\right)\right).
$

For the case when $A$ is indefinite, we use a different pair of $H_0$ and $H_1$ as discussed in Appendix~\ref{sec:apdx_matrix_dilation}. \REV{The generalization to non-Hermitian matrices is the same as for Theorem~\ref{thm:qlsp}, and it can be found in Appendix~\ref{sec:apdx_matrix_dilation} as well.} All other procedures are almost exactly the same. We summarize the results in the following theorem:

\begin{thm}
\label{thm:qlsp_zeno}
 $A$ is a $d$-sparse matrix \REV{whose singular values are in $[1/\kappa,1]$} and can be queried through oracles $O_{A,1}$ and $O_{A,2}$ in (\ref{eq:oracles_A}), and $\ket{b}$ is given by an oracle $O_B$. Then  $\ket{x}\propto A^{-1}\ket{b}$ can be obtained with fidelity $1-\epsilon$, \REV{succeeding with probability $\Omega(1)$ with ancilla qubits measurement outcome indicating success,} using
 
 1. $\Or\left(d\kappa\left(\log(\kappa)\log\log(\kappa)+\log(\frac{1}{\epsilon})\right)\right)$ queries to $O_{A,1}$, $O_{A,2}$, and $O_B$,
 
 2. $
\Or\left(nd\kappa\left(\log(\kappa)\log\log(\kappa)+\log(\frac{1}{\epsilon})\right)\right)
$ other primitive gates,

 3. $\Or(n)$ qubits.
\end{thm}

\REV{The reason we put requirement on the singular values of $A$ instead of its eigenvalues is stated in Remark~\ref{rem:dilation_non_hermitian}.}
Just like in the case of AQC-based QLSP algorithm, here if we have $\Or(\text{poly}(n))$ gate complexity for the oracles $O_{A,1}$, $O_{A,2}$, and $O_B$, then the total gate complexity will be $\wt{\Or}(\text{poly}(n)d\kappa\log(1/\epsilon))$. \REV{Although we use $\Or(n)$ qubits in total, the extra ancilla qubits we introduce in this method is in fact only $\Or(1)$. This is a further improvement from the $\Or(\log(d\kappa))$ ancilla qubits in the AQC-based QLSP algorithm.}

\REV{We remark that there is the possibility to further slightly improve  by a $\log(\kappa)$ factor (ignoring $\log\log$ terms) the asymptotic complexity  
of our QZE-based QLSP solver by using the fixed-point amplitude amplification to go from $\ket{x(f_j)}$ to $\ket{x(f_{j+1})}$ for each $j$, as discussed in \cite[Corollary 1]{wocjan2008speedup}. The bounds in this paper for many constant factors involved, particular those used in estimating the success probability of the QZE-based QLSP solver, are rather loose. However this does not concern us very much because we care mainly about the asymptotic complexity. Tighter estimates can be helpful for the actual implementation of our methods. }

\section{Discussion}\label{sec:conclusion}

In this paper, we have developed a quantum eigenstate filtering algorithm based on quantum signal processing (QSP). Our algorithm achieves the optimal query complexity among all polynomial-based eigenstate filtering methods, and uses a minimal amount of ancilla qubits. We demonstrate the usage of the eigenstate filtering method to solve quantum linear system problems (QLSP) with near-optimal complexity with respect to both the condition number $\kappa$ and the accuracy $\epsilon$. In the case when the precise value of $\kappa$ is not known \textit{a priori}, the knowledge of an upper bound of $\kappa$ would suffice.

The problem of directly  targeting at the solution $A^{-1}\ket{b}  $ is that a $(\beta,m,\epsilon)$ block-encoding of $A^{-1}$ requires at least $\beta\ge\kappa$ to make sure that $\norm{A^{-1}/\beta} \le 1$. Therefore the probability of success \REV{in the worst case} is already $\Omega(\kappa^{-2})$, and the number of \REV{rounds of amplitude amplification} needed is already $\Or(\kappa)$. \REV{Therefore to achieve near-optimal complexity, this approach can only query the block-encoding of $A$ for $\Or(\operatorname{polylog}(\kappa))$ times. To our best knowledge, there is no known method to achieve this for general matrices. However this might be possible for matrices with special structures and will be studied in future work.} 

Motivated by the success of AQC, our algorithm  views QLSP as an eigenvalue problem, which can be implemented via $P \ket{\wt{x}_0}$, where $P$ is an approximate projection operator, and $P \ket{\wt{x}_0}$ encodes the solution $\ket{x}$. The advantage of such a filtering procedure is that $P$ is a projector and $\norm{P}=1$. Hence its $(\beta,m,\epsilon)$ block-encoding only requires $\beta\sim\Or(1)$. Therefore assuming $\Or(1)$ overlap between $\ket{\wt{x}_0}$ and the solution vector, which can be satisfied by running the time-optimal AQC\ to constant precision, the probability of success of the filtering procedure is already $\Omega(1)$ without any amplitude amplification procedure. This accelerates the query complexity of the recently developed time-optimal AQC from $\wt{\Or}(\kappa/\epsilon)$ to $\wt{\Or}(\kappa\log(1/\epsilon))$. The efficient gate-based implementation of AQC still requires a time-dependent Hamiltonian simulation procedure (shown in Appendix \ref{sec:apdx_aqc}).
We then demonstrate that the dependence on the time-dependent Hamiltonian simulation procedure can be removed, using an algorithm based on the quantum Zeno effect, and the complexity is  $\wt{\Or}(\kappa\log(1/\epsilon))$. Both algorithms have constant probability of success, and can prepare the solution in terms of a pure state.



It is worth noting that the eigenstate filtering method developed in this paper works only for the case when the eigenvalue corresponding to the desired eigenstate is known exactly, which is satisfied in the eigenvalue formulation of QLSP. In order to implement the QSP-based eigenstate filtering procedure, one still needs to find the phase factors associated with the block encoding $\wt{U}$. For a given polynomial  $R_\ell(\cdot,\Delta)$, the phase factors are obtained on a classical computer in time that is polynomial in the degree and the logarithm of precision \cite[Theorems 3-5]{GilyenSuLowWiebe2018Long}. However, this procedure requires solution of all roots of a high degree polynomial, which can be unstable  for the range of polynomials $\ell\sim 100$ considered here. The stability of such procedure has recently been improved by Haah \cite{Haah2019}, though the number of bits of precision needed still scales as $\Or(\ell \log(\ell/\epsilon))$. \REV{Significant progress has been achieved recently, enabling robust computation of phase factors for polynomials of degrees ranging from thousands to tens of thousands \cite{chao2020finding,DongMengWhaleyLin2020}.}
We note that these phase factors in the eigenvalue filtering procedure only depend on $\wt{\Delta}$ and $\ell$, and therefore can be reused for different matrices once they are obtained on a classical computer. 




\section*{Acknowledgements}
 This work was partially supported by the Department of Energy under Grant No. {DE}-{SC0017867}, the Quantum Algorithm Teams Program under Grant No. DE-AC02-05CH11231, the Google Quantum Research Award (L.L.), and by the Air Force Office of Scientific Research under award number FA9550-18-1-0095 (L.L. and Y.T.). We thank Dong An, Yulong Dong, Nathan Wiebe for helpful discussions. \REV{We also thank the anonymous reviewers for helpful suggestions on improving the presentation of this paper and the applications of eigenstate filtering discussed at the end of Section~\ref{sec:eigstate_filtering}.}

\bibliographystyle{abbrvnat}
\bibliography{qlsp}

\clearpage
\pagebreak
\widetext
\appendix
\setcounter{equation}{0}
\setcounter{figure}{0}
\setcounter{table}{0}
\makeatletter
\renewcommand{\theequation}{A\arabic{equation}}
\renewcommand{\thefigure}{A\arabic{figure}}


\section{Block-encoding}
\label{sec:apdx_block_encoding}


The technique of block-encoding has been recently  discussed extensively \cite{GilyenSuLowEtAl2019, LowChuang2019}. Here we discuss how to construct block-encoding for $H-\lambda I$ which is used in eigenstate filtering, and $Q_b$, $H_0$, and $H_1$ which are used in QLSP and in particular the Hamiltonian simulation of AQC. We first introduce a simple technique we need to use repeatedly.

Given $U_A$, an $(\alpha,\REV{m},0)$-block-encoding of $A$ where $\alpha>0$, we want to construct a block encoding of $A+cI$ for some $c\in\CC$. This is in fact a special case of the linear combination of unitaries (LCU) technique introduced in \cite{ChildsKothariSomma2017}. Let 
\[
Q=\frac{1}{\sqrt{\alpha+|c|}}\left(
\begin{array}{cc}
\sqrt{|c|} & -\sqrt{\alpha} \\
\sqrt{\alpha} & \sqrt{|c|}
\end{array}
\right)
\]
and $\ket{q}=Q\ket{0}$. Since $(\bra{0^m}\otimes I) U_A (\ket{0^m}\otimes I) = A/\alpha$, we have
\[
(\bra{q}\bra{0^m}\otimes I) (\ket{0}\bra{0}\otimes e^{i\theta}I + \ket{1}\bra{1}\otimes U_A) (\ket{q}\ket{0^m}\otimes I) = \frac{1}{\alpha+|c|}(A+cI),
\]
where $\theta = \mathrm{arg} (c)$. Therefore Fig.~\ref{fig:circuit_shift} gives an $(\alpha+|c|,\REV{m}+1,0)$-block-encoding of $e^{-i\theta}(A+cI)$.

\begin{figure}[ht!]
\centering
\includegraphics[width=0.4\textwidth]{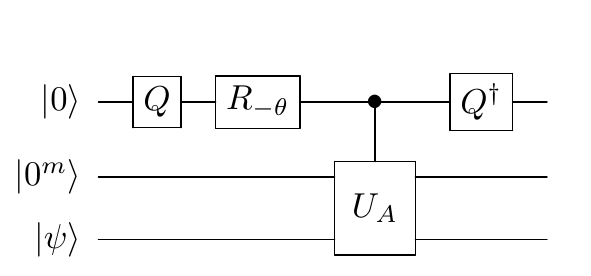}
\caption{Quantum circuit for block-encoding of $e^{-i\theta}(A+cI)$, where $c=e^{i\theta}|c|$. $R_{-\theta}\REV{=\ket{0}\bra{0}+e^{-i\theta}\ket{1}\bra{1}}$ is a phase shift gate. The three registers are the ancilla qubit for $Q$ and $\ket{q}$, the ancilla register of $U_A$, and the main register, respectively.}
\label{fig:circuit_shift}
\end{figure}

Therefore we may construct an $(\alpha+|\lambda|,m+1,0)$-block-encoding of $H-\lambda I$. We remark that since $\lambda \in \RR$, \REV{we can replace the phase shift gate with a Pauli-$Z$ gate when $\lambda>0$}. This is at the same time a $(1,m+1,0)$-block-encoding of $\wt{H}=(H-\lambda I)/(\alpha+|\lambda|)$.

Now we construct a block-encoding of $Q_b=I-\ket{b}\bra{b}$ with $\ket{b}=O_B\ket{0}$. 
Let $S_0=I-2\ket{0^{\REV{n}}}\bra{0^{\REV{n}}}$ be the reflection operator about the hyperplane orthogonal to $\ket{0^{\REV{n}}}$. Then $S_b = O_B S_0 O_B^\dagger = I-2\ket{b}\bra{b}$ is the reflection about the hyperplane orthogonal to $\ket{b}$. Note that $Q_b = (S_b+I)/2$. Therefore we can use the technique illustrated in Fig.~\ref{fig:circuit_shift} to construct a $(1,1,0)$-block-encoding of $Q_b$. Here $\ket{q}=\ket{+}=\frac{1}{\REV{\sqrt{2}}} (\ket{0}+\ket{1})$. Since $H_0 = \sigma_x \otimes Q_b$, we naturally obtain a $(1,1,0)$-block-encoding of $H_0$. We denote the block-encoding as $U_{H_0}$.

For the block-encoding of $H_1$, first note that
\[
H_1 = \left(
\begin{array}{cc}
I & 0 \\
0 & Q_b 
\end{array}
\right)
\left(
\begin{array}{cc}
0 & A \\
A & 0 
\end{array}
\right)
\left(
\begin{array}{cc}
I & 0 \\
0 & Q_b 
\end{array}
\right).
\]
From the block-encoding of $Q_b$, we can construct the block-encoding of controlled-$Q_b$ by replacing all gates with their controlled counterparts. The block matrix in the middle is $\sigma_x \otimes A$. For a $d$-sparse matrix $A$, we have a $(d,n+2,0)$-block-encoding of $A$, and therefore we obtain a $(d,n+2,0)$-block-encoding of $\sigma_x \otimes A$. Then we can use the result for the product of block-encoded matrix  \cite[Lemma 30]{GilyenSuLowEtAl2019} to obtain a $(d,n+4,0)$-block-encoding  of $H_1$, denoted by $U_{H_1}$.

\REV{The block-encodings of $H_0$ and $H_1$ allow us to block-encode linear combinations of them as well. We need access to $H(f)=(1-f)H_0+fH_1$ which is used extensively in Section~\ref{sec:quantum_zeno}. This is done through \cite[Lemma 29]{GilyenSuLowEtAl2019}. When applying the lemma we need the state preparation pair $(P_L,P_R)$ such that
\[
P_L\ket{0} = P_R\ket{0} = \frac{1}{\sqrt{1-f+fd}}(\sqrt{1-f}\ket{0}+\sqrt{fd}\ket{1}).
\]
The presence of the factor $d$ is because $H_1$ is subnormalized by a factor of $d$ in its block-encoding. By this lemma we obtain a $(1-f+fd,n+6,0)$-block-encoding of $H(f)$. Here $1-f+fd$ comes from the normalizing factor in the state preparation pair, and $n+6$ is the sum of the numbers of ancilla qubits used in the block-encodings of $H_0$ and $H_1$, plus one additional qubit used for the state preparation pair.}

\section{Implementing the reflection operator and $\theta$-reflection operator}
\label{sec:apdx_ref_and_theta_ref}

\REV{In this appendix we prove Theorem~\ref{thm:reflection} and Corollary~\ref{cor:theta_reflection} by constructing the quantum circuits. In both the theorem and the corollary we assume, as in Theorem~\ref{thm:eigenstate_filter}, that $H$ is a Hermitian matrix and $U_H$ is an $(\alpha,m,0)$-block-encoding of $H$. Also $\lambda$ is an eigenvalue of $H$ that is separated from the rest of the spectrum by a gap $\Delta$.}

\REV{We first prove Theorem~\ref{thm:reflection} by constructing the circuit for the reflection operator
\[
R_\lambda = 2P_\lambda-I,
\]
where $P_\lambda$ is the projection operator into the $\lambda$-eigenspace of $H$.
To do this we use the following polynomial
\[
S_\ell(x;\delta) = \frac{2R_\ell(x;\delta)-1}{\max_{y\in [-1,1]}|2R_\ell(y;\delta)-1|}.
\]
The first thing we should notice about this polynomial is that it is even and therefore can be implemented via QSP by Theorem~\ref{thm:qsp_parity}. The normalization is done so that we have $|S_\ell(x;\delta)|\leq 1$ for all $x\in[-1,1]$. Because $-\epsilon\leq \min_{y\in [-1,1]}R_\ell(y;\delta)<0$ and $\max_{y\in[-1,1]}R_\ell(y;\delta)=1$,
we have 
\[
1\leq \max_{y\in [-1,1]}|2R_\ell(y;\delta)-1| \leq 1+2\epsilon.
\]
Therefore
\begin{equation}
\label{eq:bound_Sl_1}
-1-2\epsilon\leq S_\ell(x;\delta) \leq \frac{-1+2\epsilon}{1+2\epsilon}\leq -1+4\epsilon,\quad x\in\mathcal{D}_{\delta},
\end{equation}
and
\begin{equation}
\label{eq:bound_Sl_2}
1-2\epsilon\leq \frac{1}{1+2\epsilon}\leq S_\ell(0;\delta) \leq 1.
\end{equation}
Now for $H$, we define $\wt{H}=(H-\lambda I)/(\alpha+|\lambda|)$ and $\wt{\Delta}=\Delta/2\alpha$ as done in the proof of of Theorem~\ref{thm:eigenstate_filter}. Then applying the polynomial $S_\ell(x;\wt{\Delta})$ to $\wt{H}$, because all eigenvalues of $\wt{H}$ are contained in $\mathcal{D}_{\wt{\Delta}}\cup\{0\}$, they are mapped to either close to $1$ or close to $-1$. Thus by Eqs.~\eqref{eq:bound_Sl_1} and \eqref{eq:bound_Sl_2} we have
\[
\|S_\ell(\wt{H};\wt{\Delta})-R_\lambda\|\leq 4\epsilon.
\]
Since $S_\ell(x;\wt{\Delta})$ is a real even polynomial that takes value in $[-1,1]$ when $x\in[-1,1]$, we can implement a $(1,m+2,0)$-block-encoding of $S_\ell(\wt{H};\wt{\Delta})$ through QSP by Theorem~\ref{thm:qsp_parity}. We denote this block-encoding by $\mathcal{U}_R$. We have
\[
\|(\bra{0^{m+2}}\otimes I)\mathcal{U}_R (\ket{0^{m+2}}\otimes I) - R_\lambda\|=\|S_\ell(\wt{H};\wt{\Delta}) - R_\lambda\| \leq 4\epsilon.
\]
Therefore $\mathcal{U}_R$ is an $(1,m+2,4\epsilon)$-block-encoding of $R_\lambda$. Thus we have proved Theorem~\ref{thm:reflection}.
}

\REV{We then prove Corollary~\ref{cor:theta_reflection} by constructing a block-encoding of the $\theta$-reflection operator
\[
P_\lambda + e^{\I \theta}(I-P_\lambda).
\]
One might be tempted to directly find a polynomial to approximate this matrix function. However such a polynomial would have complex coefficients, and we would need to apply QSP to the real and imaginary parts separately. This in turn needs an extra LCU step to add the two parts up, resulting in reduced success probability. Therefore instead of using a new polynomial, we use the block-encoding $\mathcal{U}_R$ we have already constructed, and then apply a 1-bit phase estimation on it. This enables u s to distinguish between the $\lambda$-eigenspace and its orthogonal complement, since all the eigenvalues of $R_\lambda$ are either $1$ or $-1$. We then apply the phase factor $e^{\I\theta}$ only to the correct subspace. Finally we uncompute the additional ancilla qubit. The circuit takes the following form, as shown in Figure~\ref{fig:theta_ref}:
}

\begin{figure}[h]
    \centering
    \includegraphics[width=0.7\textwidth]{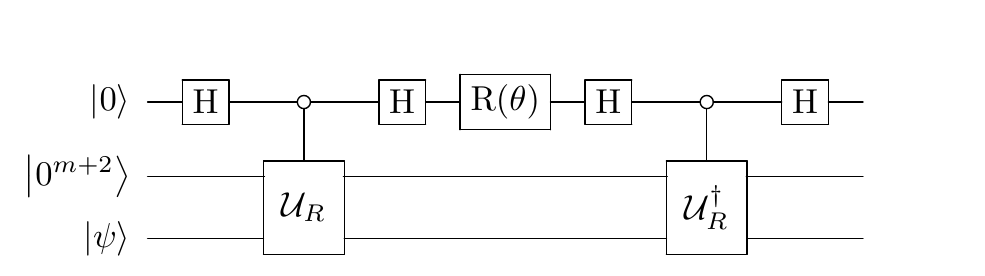}
    \caption{\REV{The quantum circuit for implementing the $\theta$-reflection operator. $\mathrm{H}$ is the Hadamard gate and $\mathrm{R}(\theta)=\ket{0}\bra{0}+e^{\I\theta}\ket{1}\bra{1}$ is the phase-shift gate.}}
    \label{fig:theta_ref}
\end{figure}

\REV{
We introduced one additional ancilla qubit in the initial state $\ket{0}$, and the second register in the above circuit is for the ancilla qubits in Theorem~\ref{thm:reflection}. The last register is the main register prepared in the state $\ket{\psi}$ on which we want to apply the operator $P_\lambda+e^{\I\theta}(I-P_\lambda)$. Thus we have proved Corollary~\ref{cor:theta_reflection}.
}

\section{Gate-based implementation of time-optimal adiabatic quantum computing}
\label{sec:apdx_aqc}

\REV{In Theorem~\ref{thm:qlsp} we used an adiabatic time evolution to prepare an initial state for eigenstate filtering. In this appendix we discuss how to implement this time evolution on a gate-based quantum computer.}
Consider the adiabatic evolution
$$
\frac{1}{T}\I \partial_s \left|\psi_T(s)\right> = H(f(s))\left|\psi_T(s)\right>, \quad                        \ket{\psi_T(0)}=\ket{0}\ket{b},
$$
Where $H(f)=(1-f)H_0+fH_1$ for $H_0$ and $H_1$ defined in (\ref{eq:H0}) and (\ref{eq:H1}). It is proved in \cite{AnLin2019,SubasiSommaOrsucci2019} that the gap between $0$ and the rest of the eigenvalues of $H(f)$ is lower bounded by $1-f+f/\kappa$. With this bound  the scheduling \eqref{eqn:AQCSchedule_explicit} in the AQC(p) scheme results in $\mathcal{O}(\kappa/\epsilon)$ runtime complexity to solve QLSP. As mentioned before, the fact that the $0$-eigenspace of $H(f(s))$ is two dimensional is not a problem because $\ket{\psi_T(t)}$ is orthogonal to $\ket{1}\ket{b}$ for all $t$.

In order to carry out AQC\ efficiently using a gate-based implementation, we use the recently developed time-dependent Hamiltonian simulation method based on truncated Dyson series introduced in \cite{LowWiebe2018}. In Hamiltonian simulation, several types of input models for the Hamiltonian are in use. Hamiltonians can be input as a linear combination of unitaries \cite{BerryChildsCleveEtAl2015}, using its sparsity structure \cite{AharonovTaShma2003,LowChuang2017}, or using its block-encoding \cite{LowChuang2019,LowWiebe2018}. For a time-dependent Hamiltonian Low and Wiebe designed an input model based on block-encoding  named HAM-T \cite[Definition 2]{LowWiebe2018}, as a block-encoding of $\sum_s\ket{s}\bra{s}\otimes H(s)$ where $s$ is a time step and $H(s)$ is the Hamiltonian at this time step.

In the gate-based implementation of the time-optimal AQC, we construct HAM-T in Fig.~\ref{fig:circuit_ham_t}. We need to use the block-encodings $U_{H_0}$ and $U_{H_1}$ introduced in Appendix \ref{sec:apdx_block_encoding}, which requires $n_0=1$ and $n_1=n+4$ ancilla qubits, respectively.
Our construction of HAM-T satisfies
\begin{equation}
(\bra{s} \bra{0^{l+1+n_0}} \otimes I \otimes \bra{0^{n_1+1}}) \text{HAM-T} 
(\ket{s} \ket{0^{l+1+n_0}} \otimes I \otimes \ket{0^{n_1+1}}) = H(f(s))/d,
\end{equation}
for any $s\in \mathcal{S}=\{j/2^l:j=0,1,\ldots,2^l-1\}$. 
\begin{figure}[ht!]
\centering
\includegraphics[width=0.6\textwidth]{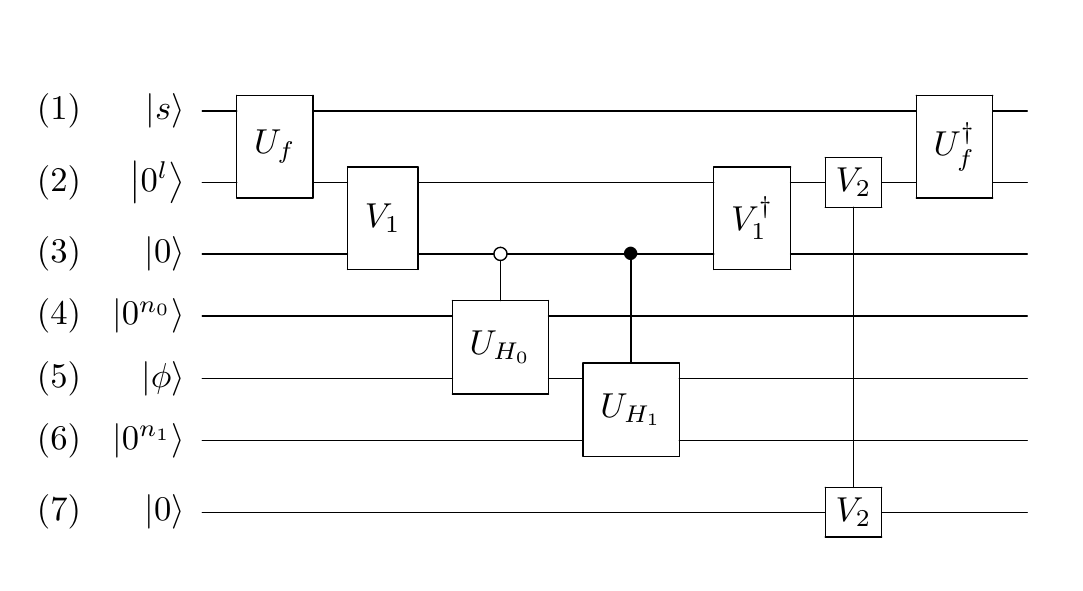}
\caption{Quantum circuit for HAM-T. The registers from top to bottom are: (1) input register for $s$ (2) register for storing $f(s)$ (3) register for a control qubit (4) ancilla register for $U_{H_0}$ (5) main register for input state $\ket{\phi}$ (6) ancilla register for $U_{H_1}$ (7) register for changing normalizing factor from $\alpha(s)$ to $d$. }
\label{fig:circuit_ham_t}
\end{figure}

In this unitary HAM-T we also need the unitary
\begin{equation}
U_{f}\ket{s}\ket{z} = \ket{s}\ket{z\oplus f(s)}
\end{equation}
to compute the scheduling function needed in the time-optimal AQC, and the unitaries
\begin{equation}
\begin{aligned}
V_1 &= \sum_{s\in\mathcal{S}} \ket{s}\bra{s}\otimes\frac{1}{\sqrt{1-s+ds}}\left(
\begin{array}{cc}
\sqrt{1-s} & -\sqrt{ds} \\
\sqrt{ds} & \sqrt{1-s}
\end{array}
\right) \\
V_2 &= \sum_{s\in\mathcal{S}} \ket{s}\bra{s}\otimes\left(
\begin{array}{cc}
\frac{\alpha(s)}{d} & -\sqrt{1-\left(\frac{\alpha(s)}{d}\right)^2} \\
\sqrt{1-\left(\frac{\alpha(s)}{d}\right)^2} & \frac{\alpha(s)}{d}
\end{array}
\right), \\
\end{aligned}
\end{equation}
where $\alpha(s) = 1-s+ds$. Here $V_1$ is used for preparing the linear combination $(1-f(s))U_{H_0}+f(s)U_{H_1}$. Without $V_2$ the circuit would be a $(\alpha(s),l+n_0+n_1+2,0)$-block-encoding of $\sum_s \ket{s}\bra{s}\otimes H(s)$, but with $V_2$ it becomes a $(d,l+n_0+n_1+2,0)$-block-encoding, so that the normalizing factor is time-independent, as is required for the input model in \cite{LowWiebe2018}.

For the AQC with positive definite $A$ we have $n_0=1$ and $n_1=n+4$. For the Hermitian indefinite case we have $n_0=2$ and $n_1=n+4$. The increase of $n_0$ from $1$ to $2$ is due to the additional operation of linear combination of matrices. For $H_1$ we can perform one less matrix-matrix multiplication, and hence the value of $n_1$ remains unchanged (see Appendix \ref{sec:apdx_matrix_dilation}).

Following \cite[Corollary 4]{LowWiebe2018}, we may analyze the different components of costs in the Hamiltonian simulation of AQC. For time evolution from $s=0$ to $s=1$,  HAM-T is a $(dT,l+n_0+n_1+2,0)$-block-encoding of $\sum_s \ket{s}\bra{s}\otimes TH(s)$. With the scheduling function given in \cite{AnLin2019} we have $\|TH(s)\|=\mathcal{O}(Td)$ and $\|\frac{\mathrm{d(}TH(s))}{\mathrm{d}s}\|=\mathcal{O}(dT\kappa^{p-1})$. We choose $p=1.5$ and by \cite[Theorem 1]{AnLin2019} we have $T=\mathcal{O}(\kappa)$. We only need to simulate up to constant precision, and therefore we can set $l=\mathcal{O}(\log(d\kappa))$. The costs are then 
\begin{enumerate}
\item Queries to HAM-T: $\mathcal{O}\left(d\kappa\frac{\log(d\kappa)}{\log\log(d\kappa)}\right)$,
\item Primitive gates: $\mathcal{O}\left(d\kappa(n+\log(d\kappa))\frac{\log(d\kappa)}{\log\log(d\kappa)}\right)$,
\item Qubits: $\mathcal{O}(n+\log(d\kappa))$.
\end{enumerate}

\section{The matrix dilation method}
\label{sec:apdx_matrix_dilation}

\REV{In Theorem~\ref{thm:qlsp} and Theorem~\ref{thm:qlsp_zeno},} in order to extend the time-optimal AQC method, \REV{and the QZE-based method} to Hermitian indefinite matrices, we follow \cite[Theorem 2]{AnLin2019}, where $H_0$ and $H_1$\REV{, as constructed in Ref.~\cite{SubasiSommaOrsucci2019}} are given by
\begin{equation}
\begin{aligned}
H_0 &= \sigma_{+}\otimes [(\sigma_z \otimes I_N)Q_{+,b}] + \sigma_{-}\otimes [Q_{+,b}(\sigma_z \otimes I_N)], \\
H_1 &= \sigma_{+}\otimes [(\sigma_{\REV{x}} \otimes A)Q_{+,b}] + \sigma_{-}\otimes [Q_{+,b}(\sigma_{\REV{x}} \otimes A)].
\end{aligned}
\label{eqn:Hdilation}
\end{equation}
Here $\sigma_{\pm}=(\sigma_x \pm i\sigma_{y})\REV{/2}$ and $Q_{+,b}=I_{\REV{2}N} - \ket{+}\ket{b}\bra{+}\bra{b}$. The dimension of the dilated matrices $H_0,H_1$ is $4N$. The lower bound for the gap of $H(f)$ then becomes $\sqrt{(1-f)^2+f^2/\kappa^{\REV{2}}}$ \cite{SubasiSommaOrsucci2019}. However in order to simplify our analysis we give a weaker lower bound
\[
\Delta_*(f) = \frac{1}{\sqrt{2}}\left(1-f+\frac{f}{\kappa}\right),
\] 
\REV{which differs from the gap lower bound in \eqref{eq:gap_aqc} by a factor of $\sqrt{2}$.} The initial state is $\ket{0}\ket{-}\ket{b}$\REV{, where $\ket{-}=\frac{1}{\sqrt{2}}(\ket{0}-\ket{1})$,} and the goal is to obtain $\ket{0}\ket{+}\ket{x}$. \REV{In the AQC-based QLSP solver,} after running the AQC we can remove the second qubit by measuring it with respect to the $\{\ket{+},\ket{-}\}$ basis and accepting the result corresponding to $\ket{+}$. The resulting query complexity remains unchanged. We remark that the matrix dilation here is only needed for AQC. The eigenstate filtering procedure can still be applied to the original matrix of dimension $2N$. \REV{The same is true for the QZE-based method.}

For a general matrix, we may first consider the extended linear system. Define   
 an extended QLSP $\mathfrak{A}\ket{\mathfrak{x}} = \ket{\mathfrak{b}}$
in dimension $2N$ where
\begin{equation*}
    \mathfrak{A} = \sigma_+ \otimes A + \sigma_- \otimes A^{\dagger}
    =\left(\begin{array}{cc}
        0 & A \\
        A^\dagger & 0
    \end{array}\right), \quad  \ket{\mathfrak{b}} = \REV{\ket{0,b}}.
\end{equation*}
Here $\mathfrak{A}$ is a Hermitian matrix of dimension $2N$, 
with condition number $\kappa$ and $\|\mathfrak{A}\| = 1$, 
and $\ket{\mathfrak{x}} = \ket{1,x}$ solves the extended QLSP. 
Therefore the time-optimal AQC \REV{and the QZE procedure} can be applied to the Hermitian matrix $\mathfrak{A}$
to prepare an $\epsilon$-approximation of $x$. 
The dimension of the corresponding $H_0,H_1$ matrices is $8N$. Again the matrix dilation method used in Eq. \eqref{eqn:Hdilation} is not needed for the eigenstate filtering step. 


\section{Optimality of the Chebyshev filtering polynomial}
\label{sec:apdx_optimal_poly}

In this section we prove Lemma~\ref{lem:minimax_poly}. We define 
\[
Q_\ell(x;\Delta) = T_{\ell}\left(-1+2\frac{x^2-\Delta^2}{1-\Delta^2}\right),
\]
then $R_{\ell}(x;\Delta)=Q_{\ell}(x;\Delta)/Q_{\ell}(0;\Delta)$. \REV{Here $T_\ell(x)$ is the $\ell$-th Chebyshev polynomial of the first kind and $0<\Delta<1$.} We need to use the following lemma\REV{, which is similar to the well-known result discussed in \cite[Proposition 2.4]{sachdeva2013faster}, \cite[Theorem 6.25]{Saad2003}, and \cite[Theorem 7]{erdos1947some}}:
\begin{lem}
\label{lem:minimax_poly_2}
For any $p(x)\in\mathbb{P}_{2\ell}[x]$ satisfying $|p(x)|\leq1$ for
all $x\in\mathcal{D}_{\Delta}$, \REV{where $\mathcal{D}_\Delta=[-1,-\Delta]\cup[\Delta,1]$,} $|Q_{\ell}(x;\Delta)|\geq|p(x)|$ for all $x\notin\mathcal{D}_{\Delta}$.\end{lem}
\begin{proof}
We prove by contradiction. If there
exists $q(x)\in\mathbb{P}_{2\ell}[x]$ such that $|q(x)|\leq1$ for all
$x\in\mathcal{D}_\Delta$ and there exists $y\notin\mathcal{D}_\Delta$ such that
$|q(y)|>|Q_{\ell}(\REV{y};\Delta)|$, then letting $h(x)=Q_{\ell}(x;\Delta)-q(x)\frac{Q_{\ell}(y;\Delta)}{q(y)}$,
we want to show $h(x)$ has at least $2\ell+1$ distinct zeros.

First note that there exist $-1=y_{1}<y_{2}<\cdots<y_{\ell+1}=1$ such
that $|T_{\ell}(y_{j})|=1$, and $T_{\ell}(y_{j})T_{\ell}(y_{j+1})=-1$. Therefore
there exist $\Delta=x_{1}<x_{2}<\cdots<x_{\ell+1}=1$ such that $|Q_{\ell}(\pm x_{j};\Delta)|=1$,
and $Q_{\ell}(x_{j};\Delta)Q_{\ell}(x_{j+1};\Delta)=-1$. In other words, $Q_{\ell}(\cdot;\Delta)$ maps
each $(x_{j},x_{j+1})$ and $(-x_{j+1},-x_{j})$ to $(-1,1)$, and
the mapping is bijective for each interval. Because $|\frac{Q_{\ell}(y;\Delta)}{q(y)}|<1$,
there exists $z_{j},w_{j}\in(x_{j},x_{j+1})$ for each $j$ such that
$h(z_{j})=h(-w_{j})=0$. Therefore $\{z_{j}\}$ and $\{-w_{j}\}$
give us $2\ell$ distinct zeros. Another zero can be found at $y$ as
$h(y)=Q_{\ell}(y)-Q_{\ell}(y)=0$. Therefore there are $2\ell+1$ distinct
zeros. 

However $h(x)$ is of degree at most $2\ell$. This shows $h(x)\equiv0$.
This is clearly impossible since $h(1)=Q_{\ell}(1;\Delta)-q(1)\frac{Q_{\ell}(y;\Delta)}{q(y)}=1-q(1)\frac{Q_{\ell}(y;\delta)}{q(y)}>0$. 
\end{proof}

\REV{
Lemma~\ref{lem:minimax_poly_2} shows that for any $y\notin\mathcal{D}_\Delta$,
\[
\max_{\substack{p(x)\in\mathbb{P}_{2\ell}[x] \\ |p(x)|\leq 1,\forall x\in \mathcal{D}_\Delta}} |p(y)|=|Q_\ell(y;\Delta)|.
\]
This is equivalent to
\[
\max_{\substack{p(x)\in\mathbb{P}_{2\ell}[x] }} \frac{|p(y)|}{\max_{x\in \mathcal{D}_\Delta}|p(x)|}=|Q_\ell(y;\Delta)|,
\]
which is in turn equivalent to
\[
\min_{\substack{p(x)\in\mathbb{P}_{2\ell}[x] }} \frac{\max_{x\in \mathcal{D}_\Delta}|p(x)|}{|p(y)|}=\frac{1}{|Q_\ell(y;\Delta)|},
\]
and
\[
\min_{\substack{p(x)\in\mathbb{P}_{2\ell}[x] \\ |p(y)|\leq 1 }} \max_{x\in \mathcal{D}_\Delta}|p(x)|=\frac{1}{|Q_\ell(y;\Delta)|}.
\]
}
This implies (i) of Lemma~\ref{lem:minimax_poly}\REV{: we only need to set $y=0$ and observe that 
\[
 \max_{x\in \mathcal{D}_\Delta}|R_{\ell}(x;\Delta)|= \frac{1}{|Q_\ell(0;\Delta)|},
\]
since the Chebyshev polynomials take value between $[-1,1]$ on the interval $[-1,1]$. From the above discussion we may derive a more general result, that $R_\ell(x;\Delta)$ solves the following minimax problem:
\[
\underset{\substack{p(x)\in\mathbb{P}_{2\ell}[x] \\ p(y)=R_\ell(y;\Delta)}}{\mathrm{minimize}}  \max_{x\in\mathcal{D}_{\Delta}}|p(x)|.
\]
}To prove (ii) \REV{of Lemma~\ref{lem:minimax_poly}}, we need to use the following lemma, \REV{which directly follows from \cite[Eq.~(6.112)]{Saad2003}}:
\begin{lem}
Let $T_{\ell}(x)$ be the $\ell$-th Chebyshev polynomial, then
\[
T_{\ell}(1+\delta)\geq\frac{1}{2}e^{\ell\sqrt{\delta}}
\]
for $0\leq\delta\leq3-2\sqrt{2}$.\end{lem}
\begin{proof}
The Chebyshev polynomial can be rewritten as $T_{\ell}(x)=\frac{1}{2}(z^{\ell}+\frac{1}{z^{\ell}})$ for $x=\frac{1}{2}(z+\frac{1}{z})$.
Let $x=1+\delta$, then $z=1+\delta\pm\sqrt{2\delta+\delta^{2}}$.
The choice of $\pm$ does not change the value of $x$, so we choose $z=1+\delta+\sqrt{2\delta+\delta^{2}}\geq1+\sqrt{2\delta}$.
Since $\log(1+\sqrt{2\delta})\geq\sqrt{2\delta}-\delta\geq\sqrt{\delta}$ for $0\leq\delta\leq 3-2\sqrt{2}$,
we have $z^{\ell}\geq e^{\ell\sqrt{\delta}}$. Thus $T_{\ell}(x)\geq\frac{1}{2}e^{\ell\sqrt{\delta}}$.
\end{proof}

We use this lemma to prove (ii). Since $|T_\ell(-1+2\frac{-\Delta^2}{1-\Delta^2})|\geq T_\ell(1+2\Delta^2)$, when $\Delta^2\leq 1/12$, we have $2\Delta^2\leq 1/6< 3-2\sqrt{2}$. Thus by the above lemma we have $|T_\ell(-1+2\frac{-\Delta^2}{1-\Delta^2})|\geq \frac{1}{2}e^{\ell\sqrt{2\Delta^{\REV{2}}}}$. Since $|T_\ell(-1+2\frac{x^2-\Delta^2}{1-\Delta^2})|\leq 1$ for $x\in\mathcal{D}_\Delta$, we have the inequality in (ii). (iii) follows straightforwardly from the monotonicity of Chebyshev polynomials outside of $[-1,1]$.

\section{\REV{Properties of the eigenpath}}
\label{sec:apdx_eigenpath}



\REV{
In this section we construct a smooth one-parameter family of normalized quantum states $\{\ket{x(f)}\}$ satisfying Eqs.~\eqref{eq:def_xf} and \eqref{eq:geometric_phase}. $\{\ket{0}\ket{x(f)}\}$ then gives an eigenpath of the one-parameter family of Hamiltonians $\{H(f)\}$. We also prove the inequality \eqref{eq:bound_eigenpath_derivative}.
}

\REV{
We define 
\[
\ket{y(f)} = ((1-f)I+fA)^{-1}\ket{b}.
\]
Then $\braket{y(f)|y(f)}\geq 1$ because $\|(1-f)I+fA\|\leq 1$. Also $\ket{y(f)}$ is a smooth function of $f$ for $f\in(0,1)$ because $(1-f)I+fA$ is invertible in this interval, under the assumption that $A$ is Hermitian positive-definite with eigenvalues in $[1/\kappa,1]$. We construct $\ket{x(f)}$ through
\[
\ket{x(f)} = c(f)\ket{y(f)},
\]
with $c(f)$ solving the following ODE
\begin{equation}
    c'(f) = -c(f)\frac{\braket{y(f)|\partial_f|y(f)}}{\braket{y(f)|y(f)}},\quad c(0) = 1.
\end{equation}
This is a linear ODE and the right-hand side depends smoothly on $f$. Therefore the solution exists and is unique for  $f\in [0,1]$. It then follows that this construction of $\ket{x(f)}$ satisfies \eqref{eq:def_xf}. Since 
\[
\partial_f \ket{x(f)} = c'(f)\ket{y(f)} + c(f)\partial_f \ket{y(f)},
\]
we have
\[
\bra{x(f)} \partial_f \ket{x(f)} = c^*(f) [c'(f)\braket{y(f)|y(f)} + c(f)\bra{y(f)}\partial_f \ket{y(f)}]=0.
\]
Therefore Eq.~\eqref{eq:geometric_phase} is satisfied, and this in turn ensures $\ket{x(f)}$ is normalized. In this way we have constructed $\{\ket{x(f)}\}$ that satisfies all the requirements in Section~\ref{sec:quantum_zeno_eigenpath}. For $H(f)=(1-f)H_0+fH_1$, where $H_0$ and $H_1$ are defined in Eqs.~\eqref{eq:H0} and \eqref{eq:H1} respectively, we can see $H(f)\ket{\bar{x}(f)}=0$ where $\ket{\bar{x}(f)}=\ket{0}\ket{x(f)}$. Therefore $\{\ket{\bar{x}(f)}\}$ is a smooth eigenpath. 
}

\REV{
If there is another eigenpath $\{\ket{0}\ket{w(f)}\}$ satisfying $\braket{w(f)|\partial_f|w(f)}=0$, then it follows that $((1-f)I+fA)\ket{w(f)}\propto \ket{b}$. Therefore $\ket{w(f)}=e^{\I\theta(f)}\ket{x(f)}$ for some differentiable $\theta(f)$. By the geometric phase condition we can show $e^{\I\theta(f)}=1$ for all $f$ by also taking into account the initial condition $\ket{w(f)}=\ket{b}$,  and therefore $\ket{w(f)}=\ket{x(f)}$. This proves uniqueness.
}

\REV{
Now we denote by $\varepsilon_j(f)$ the eigenvalues of $H(f)$. The corresponding eigenstates are denoted by $\ket{w_j(f)}$. Because $((1-f)I+fA)\ket{x(f)}\propto \ket{b}$, we have $H(f)\ket{\bar{x}(f)}=0$. Since $\ket{x(f)}$, and as a result $\ket{\bar{x}(f)}$, is differentiable, taking derivative with respect to $f$ we have
\[
H'(f)\ket{\bar{x}(f)} + H(f)\partial_f \ket{\bar{x}(f)} = 0.
\]
Therefore
\[
\bra{w_j(f)}H'(f)\ket{\bar{x}(f)} + \bra{w_j(f)}H(f)\partial_f \ket{\bar{x}(f)} = 0.
\]
And this leads to
\[
 \bra{w_j(f)}\partial_f \ket{\bar{x}(f)} = -\frac{\bra{w_j(f)}H'(f)\ket{\bar{x}(f)}}{\varepsilon_j(f)}
\]
for any $j$ such that $\varepsilon_j(f)\neq 0$. The null space of $H(f)$ is spanned by $\ket{1}\ket{b}$ and $\ket{\bar{x}(f)}$. We have $(\bra{1}\bra{b})\ket{\bar{x}(f)}=\braket{1|0}\braket{b|x(f)}=0$, and $\braket{\bar{x}(f)|\partial_f|\bar{x}(f)}=0$ because of the geometric phase condition \eqref{eq:geometric_phase}. Since all $\ket{w_j(f)}$ such that $\varepsilon_j(f)\neq 0$, together with $\ket{1}\ket{b}$ and $\ket{\bar{x}(f)}$ form a basis of the Hilbert space, we have
\[
\partial_f \ket{\bar{x}(f)} = -\sum_{j:\varepsilon_j(f)\neq 0}\frac{\ket{w_j(f)}\bra{w_j(f)}H'(f)\ket{\bar{x}(f)}}{\varepsilon_j(f)}
\]
Therefore
\[
\begin{aligned}
\|\partial_f \ket{\bar{x}(f)}\|^2 & = \sum_{j:\varepsilon_j(f)\neq 0}\frac{|\bra{w_j(f)}H'(f)\ket{\bar{x}(f)}|^2}{\varepsilon_j^2(f)} \\
&\leq \frac{1}{\Delta_*(f)^2}\sum_{j:\varepsilon_j(f)\neq 0} |\bra{w_j(f)}H'(f)\ket{\bar{x}(f)}|^2 \\
&\leq \frac{1}{\Delta_*(f)^2}\|H'(f)\ket{\bar{x}(f)}\|^2
\end{aligned}
\]
From the definition of $H(f)$ it can be seen that $\|H'(f)\|\leq 2$. Therefore we have proved the inequality \eqref{eq:bound_eigenpath_derivative}.
}

\section{Success probability of Quantum Zeno effect QLSP algorithm}
\label{sec:apdx_success_probability}

\REV{In this appendix we rigorously prove a constant success probability lower bound for the QZE-based QLSP algorithm in Theorem~\ref{thm:qlsp_zeno}. In Section~\ref{sec:zeno_qlsp_estimates} we gave a simpler but non-rigorous proof of a constant success probability lower bound by assuming the projection for each $H(f_j)$ is done without error, i.e. $\epsilon_P=0$. Here we do not make such an assumption and show we can still find such a lower bound.}  
We will need to use the following elementary inequality, which can be easily proved using induction.
\begin{lem}
\label{lem:elementary_ineq}
If $0<a_j<1$, $0<b_j<1$, $j=0,1,2,\ldots,R-1$, then
\[
\prod_{j=0}^{\REV{R}-1} (a_j-b_j) \geq \prod_{j=0}^{\REV{R}-1} a_j - \sum_{j=0}^{\REV{R}-1} b_j.
\]
\end{lem}

\REV{We first recall the definition of the sequence of quantum states $\{\ket{x(f_j)}\}$, with each $\ket{x(f_j)}$ defined through \eqref{eq:def_xf} and \eqref{eq:geometric_phase}, satisfying $H(f_j)\ket{0}\ket{x(f_j)}=0$, and the sequence of quantum states $\{\ket{\wt{x}(f_j)}\}$, with each $\ket{\wt{x}(f_j)}$ defined recursively by \eqref{eq:approx_xf}.}
We need to use the following bound for the overlap between $\ket{x(f_j)}$ and $\ket{x(f_{j+1})}$ derived from Eqs. \eqref{eq:xf_overlap} \eqref{eq:xf_distance} and \eqref{eq:curve_length_reparam}.
\begin{equation}
\label{eq:xf_overlap_new}
|\braket{x(f_j)|x(f_{j+1})}| \geq 1-\frac{1}{2}\|\ket{x(f_{j+1})}-\ket{x(f_j)}\|^2\geq 1-\frac{2\log^2(\kappa)}{M^2(1-1/\kappa)^2}.
\end{equation}

With these tools we will first bound several overlaps in the following lemma
\begin{lem}
\label{lem:bound_overlap}
When $M\geq \REV{\frac{4\log^2(\kappa)}{(1-1/\kappa)^2}}$ and $\epsilon_P\leq \frac{1}{128}$, we have for $j=0,1,\ldots,M-1$: 

(i) $|\braket{x(f_j)|x(f_{j+1}})|\geq 1-\frac{1}{2M}$, 

(ii) $|\braket{x(f_j)|\wt{x}(f_j)}|\geq 1-4\epsilon_P$, 

(iii) $|\braket{\wt{x}(f_j)|x(f_{j+1})}|\geq 1-\frac{1}{2M}-4\epsilon_P-2\sqrt{2\epsilon_P}$.
\end{lem}

\begin{proof}
(i) derives directly from (\ref{eq:xf_overlap_new}). We then want to derive (ii) and (iii) inductively. First we have
\[
\begin{aligned}
|\braket{\wt{x}(f_j)|x(f_{j+1})}| &= |\braket{\wt{x}(f_j)|P_0(f_j)|x(f_{j+1})}
+\braket{\wt{x}(f_j)|I-P_0(f_j)|x(f_{j+1})}| \\
&\geq |\braket{x(f_j)|x(f_{j+1})}|\cdot|\braket{x(f_j)|\wt{x}(f_j)}| - \|(I-P_0(f_j))\ket{\wt{x}(f_j)}\|.
\end{aligned}
\]
Because
\[
\|(I-P_0(f_j))\ket{\wt{x}(f_j)}\|^2=1-|\braket{x(f_j)|\wt{x}(f_j)}|^2,
\]
we then have
\[
|\braket{\wt{x}(f_j)|x(f_{j+1})}| \geq |\braket{x(f_j)|x(f_{j+1})}|\cdot|\braket{x(f_j)|\wt{x}(f_j)}| - \sqrt{1-|\braket{x(f_j)|\wt{x}(f_j)}|^2}.
\]
We denote
\[
|\braket{x(f_j)|\wt{x}(f_j)}| = 1-\nu_j,
\]
then
\begin{equation}
\label{eq:overlap_type_3}
\begin{aligned}
|\braket{\wt{x}(f_j)|x(f_{j+1})}| &\geq (1-\frac{1}{2M})(1-\nu_j)-\sqrt{1-(1-\nu_j)^2} \\
&\geq 1-\frac{1}{2M}-\nu_j-\sqrt{2\nu_j}.
\end{aligned}
\end{equation}

We now bound $\nu_{j+1}$ using $|\braket{\wt{x}(f_j)|x(f_{j+1})}|$. First \REV{using the fact that $\|\wt{P}_0(f_{j})-P_0(f_{j})\|\leq\epsilon_P$ where the approximate projection operator $\wt{P}_0(f_j)$ is defined in Eq.~\eqref{eq:def_wtP0},} we have 
\begin{equation}
\label{eq:fidelity_time_j}
\begin{aligned}
|\braket{\wt{x}(f_{j+1})|x(f_{j+1})}| &=  \frac{|\braket{\wt{x}(f_{j})|\wt{P}_0(f_{j+1})|x(f_{j+1})}|}{\|\wt{P}_0(f_{j+1})\ket{\wt{x}(f_j)}\|} \\
&\geq \frac{|\braket{\wt{x}(f_{j})|P_0(f_{j+1})|x(f_{j+1})}|-\epsilon_P}{\|P_0(f_{j+1})\ket{\wt{x}(f_j)}\|+\epsilon_P} \\
&= \frac{|\braket{\wt{x}(f_j)|x(f_{j+1})}|-\epsilon_P}{|\braket{\wt{x}(f_j)|x(f_{j+1})}|+\epsilon_P}
 \\
&\geq 1-\frac{2\epsilon_P}{|\braket{\wt{x}(f_j)|x(f_{j+1})}|}. \\
\end{aligned}
\end{equation}
This leads to
\begin{equation}
\label{eq:recurrence_nu}
\nu_{j+1} \leq \frac{2\epsilon_P}{|\braket{\wt{x}(f_j)|x(f_{j+1})}|} 
\leq \frac{2\epsilon_P}{1-\frac{1}{2M}-\nu_j-\sqrt{2\nu_j}},
\end{equation}
which establishes a recurrence relation for $\nu_j$. Because $\nu_0=0$, $M\REV{\geq \frac{4\log^2(\kappa)}{(1-1/\kappa)^2}} \geq 4$ and $\epsilon_P\leq\frac{1}{128}$, we can prove inductively that $\nu_j\leq \frac{1}{32}$. Taking this into (\ref{eq:recurrence_nu}) we have
\[
\nu_{j+1}\leq 4\epsilon_P,
\]
which proves (ii). Taking this into (\ref{eq:overlap_type_3}) we have (iii).

\end{proof}

\REV{
An immediate corollary of (iii) in the above lemma is
\begin{equation}
\label{eq:overlap_bound_1/2}
|\braket{\wt{x}(f_j)|x(f_{j+1})}|\geq 1-\frac{1}{2M}-4\epsilon_P-2\sqrt{2\epsilon_P} \geq \frac{1}{2},
\end{equation}
for $j=0,1,\ldots,M-1$, $M\geq 4$, and $\epsilon_P\leq 1/128$.}
With these tools we are now ready to estimate the success probability $\psucc$. We have
\begin{equation}
\label{eq:success_prob_bound_step_1}
\begin{aligned}
\psucc &= \prod_{j=0}^{M-1}\|\wt{P}_0(f_{j+1})\ket{\wt{x}(f_j)}\|^2 \\
&\geq \left( \prod_{j=0}^{M-2}\left(\|P_0(f_{j+1})\ket{\wt{x}(f_j)}\| - \epsilon_P\right) \right)^2 \REV{\left(\|P_0(1)\ket{\wt{x}(f_{M-1})}\| - \frac{\epsilon}{4}\right)^2}  \\
&\geq \REV{\frac{1}{16}}\left( \prod_{j=0}^{M-\REV{2}}\left(\|P_0(f_{j+1})\ket{\wt{x}(f_j)}\| - \epsilon_P\right) \right)^2 \\
&\REV{\geq \frac{1}{16}\left( \prod_{j=0}^{M-1}\left(\|P_0(f_{j+1})\ket{\wt{x}(f_j)}\| - \epsilon_P\right) \right)^2} \\
&\geq \REV{\frac{1}{16}}\left( \prod_{j=0}^{M-1}\|P_0(f_{j+1})\ket{\wt{x}(f_j)}\| - M\epsilon_P \right)^2.
\end{aligned}
\end{equation}
In the last line we have used Lemma~\ref{lem:elementary_ineq}. In the second line the $j=M-1$ case is treated differently because in the last step we need to attain \REV{$\epsilon/4$} precision \REV{for eigenstate filtering. We bound the success probability of the last step using
\[
\left(\|P_0(1)\ket{\wt{x}(f_{M-1})}\| - \frac{\epsilon}{4}\right)^2 
=\left(\|\braket{x(f_M)|\wt{x}(f_{M-1})}\| - \frac{\epsilon}{4}\right)^2 
\geq \left(\frac{1}{2}-\frac{1}{4}\right)^2=\frac{1}{16},
\]
where we have used Eq.~\eqref{eq:overlap_bound_1/2} for $j=M-1$.
} 
This inequality motivates us to bound $\prod_{j=0}^{M-1}\|P_0(f_{j+1})\ket{\wt{x}(f_j)}\|$, for which, by Lemma~\ref{lem:bound_overlap}, we have
\begin{equation}
\label{eq:success_prob_component_1}
\begin{aligned}
\prod_{j=0}^{M-1}\|P_0(f_{j+1})\ket{\wt{x}(f_j)}\| &= \prod_{j=0}^{M-1}|\braket{\wt{x}(f_j)|x(f_{j+1})}| \\
&\geq \left(1-\frac{1}{2M}-4\epsilon_P-2\sqrt{2\epsilon_P}\right)^M \\
&\geq \frac{1}{2}-M(4\epsilon_P+2\sqrt{2\epsilon_P}).
\end{aligned}
\end{equation}
In Lemma~\ref{lem:bound_overlap} we have required that $\epsilon_P\leq \frac{1}{128}$ and $M\geq \REV{\frac{4\log^2(\kappa)}{(1-1/\kappa)^2}}\geq 4$. Therefore when we further require $\epsilon_P\leq \frac{1}{162M^2}$ we have
\[
\prod_{j=0}^{M-1}\|P_0(f_{j+1})\ket{\wt{x}(f_j)}\| \geq \frac{1}{4}.
\]
Substituting this into (\ref{eq:success_prob_bound_step_1}) we have
\[
\psucc \geq \REV{\frac{1}{16}}\left(\frac{1}{4}-M\epsilon_P\right)^2 \geq  \REV{\frac{1}{16}}\left(\frac{1}{4}-\frac{1}{162M}\right)^2\geq \REV{\frac{1}{400}},
\]
since $M\geq 4>1$. \REV{We remark that because we mostly only care about the asymptotic complexity we did not bound this probability very tightly, and this bound may be a very loose one. The actual success probability can be much larger than this and can be further increased by optimizing the choice of $M$ and $\epsilon_P$. }

\end{document}